%% file: main.tex
\documentclass[twoside,11pt]{article}

\usepackage{layouts}
\usepackage{epsfig}
\usepackage{amssymb}
\usepackage{natbib}
\usepackage[page]{appendix}
\usepackage{hyperref}
\usepackage{graphicx}
\usepackage{thumbpdf}
\usepackage{amsfonts,amstext,amsmath,amsthm}
\usepackage{dsfont}
\usepackage{mathtools}
\usepackage{accents}
\usepackage{color}
\usepackage{rotating}
\usepackage[utf8]{inputenc}
\usepackage{booktabs}
\usepackage{tikz}
\usepackage{tkz-graph}
\usetikzlibrary{automata, positioning}
\usetikzlibrary{decorations.pathreplacing}
\usetikzlibrary{shapes.geometric}
\usetikzlibrary{backgrounds}
\usetikzlibrary{arrows.meta}
\usepackage{multirow}
\usepackage{float}
\usepackage[labelfont=bf, width=15cm, font=small]{caption}
\usepackage{subcaption}
\usepackage[left=3cm, right=3cm, top=2.5cm, bottom=3cm]{geometry}
\usepackage[shortlabels]{enumitem}
\usepackage{algorithm2e}
\RestyleAlgo{ruled}
\usepackage{fontawesome}
\usepackage{makecell}
\usepackage{placeins}
\usepackage{authblk}

\setcitestyle{authoryear,open={(},close={)}}

\input{defs}

\begin{document}

\title{\bf Causal Change Point Detection and Localization}

\author[1]{Shimeng Huang}
\author[2]{Jonas Peters}
\author[1]{Niklas Pfister}
\affil[1]{Department of Mathematical Sciences, University of Copenhagen, Copenhagen, Denmark}
\affil[2]{Department of Mathematics, ETH Z\"urich, Z\"urich, Switzerland}

\date{\today}

\maketitle

\begin{abstract}
  \input{abstract}
\end{abstract}

\input{body}

\bibliographystyle{abbrvnat}
\bibliography{bibliography} 

\clearpage

\begin{appendices}
\input{appendix}
\end{appendices}

\clearpage

\end{document}

%% file: defs.tex
\DeclareMathOperator*{\argmin}{arg\,min}

\newcommand{\bbR}{\mathbb{R}}
\newcommand{\bbN}{\mathbb{N}}

\newcommand{\cH}{\mathcal{H}}

\newcommand{\cS}{\mathcal{S}}

\newcommand{\iid}{\stackrel{\mathrm{iid}}{\sim}}
\newcommand{\deq}{\stackrel{d}{=}}
\newcommand{\dneq}{\stackrel{d}{\neq}}
\newcommand{\bbP}{\mathbb{P}}
\newcommand{\bbE}{\mathbb{E}}
\newcommand{\bbV}{\mathbb{V}}
\newcommand{\cov}{\text{Cov}}
\newcommand{\cN}{\mathcal{N}}

\theoremstyle{definition}
\newtheorem{definition}{Definition}[section]
\newtheorem{proposition}{Proposition}[section]

\newtheorem{corollary}{Corollary}[section]
\newtheorem{lemma}{Lemma}[section]
\newtheorem{setting}{Setting}[section]

\newtheorem*{notation}{Notation}
\newtheorem*{notation2}{Additional notation}
\theoremstyle{plain}
\newtheorem{example}{Example}[section]
\theoremstyle{remark}

\newcommand{\bX}{\mathbf{X}}
\newcommand{\bY}{\mathbf{Y}}

\newcommand{\bI}{\mathbf{I}}

\newcommand{\cC}{\mathcal{C}}
\newcommand{\cD}{\mathcal{D}}
\newcommand{\cT}{\mathcal{T}}

\newcommand{\cL}{\mathcal{L}}
\newcommand{\cI}{\mathcal{I}}

\newcommand{\cK}{\mathcal{K}}

\newcommand{\fL}{\mathfrak{L}}

%% file: abstract.tex
Detecting and localizing change points in sequential data is of interest in many
areas of application. Various notions of change points have been proposed, such
as changes in mean, variance, or the linear regression coefficient. In this
work, we consider settings in which a response variable $Y$ and a set of
covariates $X=(X^1,\ldots,X^{d+1})$ are observed over time and aim to find
changes in the causal mechanism generating $Y$ from $X$. More specifically, we
assume $Y$ depends linearly on a subset of the covariates and aim to determine
at what time points either the dependency on the subset or the subset itself
changes. We call these time points causal change points (CCPs) and show that
they form a subset of the commonly studied regression change points. We propose
general methodology to both detect and localize CCPs. Although motivated by
causality, we define CCPs without referencing an underlying causal model. The
proposed definition of CCPs exploits a notion of invariance, which is a purely
observational quantity but -- under additional assumptions -- has a causal
meaning. For CCP localization, we propose a loss function that can be combined
with existing multiple change point algorithms to localize multiple CCPs
efficiently. We evaluate and illustrate our methods on simulated datasets.

%% file: body.tex
\section{Introduction}
Change point detection (i.e., testing the existence of change points) and
localization (i.e., estimating the location of change points) have been of
interest for several decades dating back to \citet{page1954continuous,
page1955test}. We consider an offline setting where we have a sequence of
independent observations $(X_1,Y_1),\ldots,(X_n,Y_n)$ with covariates
$X_i\in\mathbb{R}^{d+1}$ and a response $Y_i\in\mathbb{R}$. For all
$i\in\{1,\ldots,n\}$, denote by $\bbP_i^{X,Y}$ the joint distribution of
$(X_i,Y_i)$, which may change across $i$. We call a time point
$k\in\{2,\ldots,n\}$ a \emph{change point} if the joint distributions at time
points $k$ and $k-1$ differ, that is, if $\bbP_k^{X,Y} \neq \bbP_{k-1}^{X,Y}$
\citep[see also][]{darkhovsky1993non}. Instead of considering general change
points as defined above, one may consider a more restrictive definition of
change points, e.g., time points where there is a change in mean, variance or
conditional distribution.  Depending on the application at hand, certain types
of change points may be more relevant than others. In many applications, the
goal is to detect or localize changes in the relationship between the covariates
and the response.

In economics and other fields, ``structural changes", that is, changes in
regression models, have been extensively studied over the last few decades.
These include linear and nonlinear regression models, as well as non-parametric
regression models. Under linear regression settings, \citet{bai1996testing} and
\citet{perron2020testing} propose tests for detecting changes in the regression
parameter and the residual distribution; \citet{hansen2000testing} proposes a
test that detects changes in the regression parameter while allowing for changes
in the marginal distribution of the covariates; \citet{bai1997estimation}
considers localizing one structural change while allowing lagged and trending
covariates, and \citet{bai1997estimating} and \citet{bai1998estimating,
bai2003computation} analyze the estimation of multiple change points.
\citet{andrews1993tests} considers testing an unknown change point in part of
the parameter vector in nonlinear regression models. Testing for changes in
nonparametric regression models has been considered by,
e.g.,~\citet{orvath2002change}. In recent years, structural changes in
high-dimensional regression models have also been studied
\citep[e.g.,][]{leonardi2016computationally, wang2021statistically}. Reviews are
provided by \citet{aue2013structural}, for example, who focus on methods
detecting structural change that allow for serial dependence;
\citet{truong2020selective} consider algorithms that can be characterized by a
cost function, a search method, and a constraint on the number of changes.

By definition, structural changes refer to changes in the conditional
distribution of $Y$ given all covariates $X$. While in many applications such
changes are useful, it may also be of interest to have some type of mechanistic
understanding of the changes in order to assess their relevance. For example, if
we assume there is an underlying causal model generating the distribution over
the variables $(X, Y)$, then a structural change, that is, a change in the
conditional distribution of $Y$ given $X$, can have different causal
explanations: It could either indicate a change in causal relationship between
$Y$ and $X$, or it could merely correspond to shifts in the distribution of $X$
that do not affect the causal dependence of $Y$ on $X$. The ability to
distinguish between such changes can be useful in many applications as it allows
practitioners to pay particular attention to the more fundamental changes. In
this work, we characterize these changes based on reversing the idea of causal
invariance --- also known as autonomy or modularity
\citep[e.g.,][]{haavelmo1944probability, aldrich1989autonomy} --- which gives
the change points a causal interpretation under a causal model but is still
meaningful otherwise. The idea of causal invariance has been used in invariant
causal prediction proposed by \citet{peters2016causal} and its sequential
counterpart \citep{pfister2019invariant} for discovering the causal predictors
of a response variable, where the conditional distribution of the response
variable given its causal predictors is assumed to be unchanged across
environments (respectively, time). Our paper shows that this idea also proves
useful when detecting and localizing change points. To our knowledge, detecting
or localizing change points that can have a causal interpretation have not been
studied with one exception on detecting local causal mechanism changes in DAGs
in \citet{huang2020causal}, where the definition of the causal mechanism is
different from ours, specifically they assume that the parent set of each node
is fixed. 

\begin{notation}
We observe a sequence of independent observations $(X_1,Y_1),\ldots,(X_n,Y_n)$
with covariates $X_i\in\mathbb{R}^{d+1}$ and a response $Y_i\in\mathbb{R}$. To
avoid explicitly stating intercepts, we assume $X_i^{d+1} = 1$ for all $i
\in\{1,\ldots,n\}$, and let $\cS \coloneqq \big\{S\subseteq\{1,\ldots,d+1\}\mid
d+1\in S\big\}$. We let $\cI$ be the set of all subsets of $\{1,\ldots,n\}$ that
are sequences of consecutive indices of length greater than or equal to $2$
which we refer to as ``intervals''. For all $S\in\cS$ we denote $X_i^S \in
\bbR^{|S|}$ as the column vector of covariates $(X_i^j)_{j\in S}$ (sorted in
ascending order of the indices). We denote by $\bX\coloneqq
(X_1,\ldots,X_n)^{\top}\in\mathbb{R}^{n\times d}$ and $\bY\coloneqq
(Y_1,\ldots,Y_n)^{\top}\in\mathbb{R}^{n\times 1}$ the design matrix of the
covariates and the matrix of responses, respectively. For all $I\in\cI$, we
denote by $\bX_I$ and $\bY_I$ the submatrices formed by the rows of $\bX$ and
$\bY$ indexed by $I$, respectively (sorted in ascending order of the indices),
and additionally for all $S\in\cS$, $\bX_I^S$ denotes the submatrix of $\bX$
formed by the rows indexed by $I$ and columns indexed by $S$.
\end{notation}

This paper is organized as follows. In Section~\ref{sec:cp_and_ccp}, we define
regression change points and causal change points. Section~\ref{sec:ccpd}
focuses on the detection problem and introduces a simple procedure. In
Section~\ref{sec:ccpl}, we consider the localization problem and propose two
different methods: one that tests candidates and one that minimizes a loss
function. Numerical experiments are given in Section~\ref{sec:experiments}.

\section{Regression change points and causal change points}
\label{sec:cp_and_ccp} 

To distinguish between structural and causal changes, we first formally define
the time points of structural changes below, which we call \emph{regression
change points}.

\begin{definition}[Regression change point (RCP)] \label{def:rcp} For all $i \in
\{1,\ldots,n\}$, assume that $\bbE[X_iX_i^{\top}]$ is invertible and define the
population ordinary least squares (OLS) coefficient as
$\beta_i^{\operatorname{OLS}} \coloneqq \bbE[X_iX_i^\top]^{-1}\bbE[X_iY_i]$ and
the corresponding residual as $\epsilon_i \coloneqq Y_i -
X_i^\top\beta_i^{\operatorname{OLS}}$. Then, a time point $k \in \{2,\ldots,n\}$
is called a \emph{regression change point} (RCP) if
\begin{equation*}
\text{either} \quad 
  \beta_k^{\operatorname{OLS}}\neq\beta_{k-1}^{\operatorname{OLS}}
  \quad\text{or}\quad
  \epsilon_k \dneq \epsilon_{k-1}.
\end{equation*}
\end{definition}

While we do not assume that the conditional mean of $Y$ given $X$ is linear, the
definition of RCPs implies that if $I\in\mathcal{I}$ is an interval without an
RCP, then there exists a vector $\beta\in\mathbb{R}^{d+1}$ and a distribution
$F_{\epsilon}$ such that for all $i\in I$ it holds that
\begin{equation*}
Y_i = X_i^{\top}\beta + \epsilon \quad \text{and} \quad 
\bbE[X_i\epsilon] = 0,
\end{equation*}
with $\epsilon \sim F_{\epsilon}$
and $\beta_i^{\operatorname{OLS}}=\beta$. 

RCPs characterize changes in the conditional mean model. However, even though
these changes are sometimes interpreted as a proxy for a change in causality, it
is well-known that this interpretation can be misleading. The following example
illustrates this.

\begin{example}[RCPs in linear SCMs] \label{ex:lscm} Let $\{1,\ldots,n\}$ be
partitioned into three disjoint time intervals $I_1=\{1,\ldots,k_1-1\}$,
$I_2=\{k_1,\ldots,k_2-1\}$ and $I_3=\{k_2,\ldots,n\}$. For all
$i\in\{1,\ldots,n\}$ consider the linear structural causal model (SCM), see also
Section~\ref{sec:causal_models}, over the variables $(X_i^1, X_i^2, X_i^3,
X_i^4, Y_i)$ given by $X_i^4\coloneqq 1$ as the intercept and
\begin{subequations}
\begin{align}
X_i^{S} &\coloneqq A_i X_i + \alpha_i Y + \epsilon_i^X \label{eq:ex1cov} \\
Y_i &\coloneqq \beta_i^{\top} X_i + \epsilon_i^Y, \label{eq:ex1causal}
\end{align}
\end{subequations}
where $S = \{1,2,3\}$, $\epsilon_i^X = (\epsilon_i^{X^1}, \epsilon_i^{X^2},
\epsilon_i^{X^3})$ and $\epsilon_i^{Y}$ are jointly independent noise vectors
with mean zero, $A_i\in\mathbb{R}^{3\times 4}$, $\beta_i\in\mathbb{R}^{4}$ and
$\alpha_i\in\mathbb{R}^3$ are parameters such that the SCM induces the graphs in
Figure~\ref{fig:ex_lscm} (the intercept variable $X_i^4$ is omitted). The
specific values for the parameters $A_i$, $\beta_i$ and $\alpha_i$, as well as
the variances of $\epsilon_i^X$ and $\epsilon_i^Y$ for $i \in \{1,\ldots,n\}$
are given in Appendix~\ref{app:ex_details}.
\begin{figure}[t]
\centering
\resizebox{0.95\textwidth}{!}{
\begin{tikzpicture}
  
  \node[state] at (0, 0) (X1) {$X^1_i$};
  \node[below=of X1, state] (Y) {$Y$};
  \node[left=of Y, state] (X2) {$X^2_i$};
  \node[right=of Y, state] (X3) {$X^3_i$};
  \draw[-Latex] (X1) -- (Y);
  \draw[-Latex] (X2) -- (Y);
  \draw[-Latex] (Y) -- (X3);
  \node[] at (0, -3){$i \in I_1$};
  
  \node[state] at (6, 0) (X1) {$X^1_i$};
  \node[above right=0.05px of X1] () {\faLegal};
  \node[below=of X1, state] (Y) {$Y$};
  \node[left=of Y, state] (X2) {$X^2_i$};
  \node[right=of Y, state] (X3) {$X^3_i$};
  \node[above right=0.05px of X3] () {\faLegal};
  \draw[-Latex] (X1) -- (Y);
  \draw[-Latex] (X2) -- (Y);
  \draw[-Latex] (Y) -- (X3);
  \node[] at (6, -3){$i \in I_2$};
  
  \node[state] at (12, 0) (X1) {$X^1_i$};
  \node[below=of X1, state] (Y) {$Y$};
  \node[left=of Y, state] (X2) {$X^2_i$};
  \node[right=of Y, state] (X3) {$X^3_i$};
  \node[above right=0.05px of X3] () {\faLegal};
  \node[above right=0.05px of Y] () {\faLegal};
  \draw[-Latex] (X1) -- (X3);
  \draw[-Latex] (X2) -- (Y);
  \draw[-Latex] (Y) -- (X3);
  \node[] at (12, -3){$i \in I_3$};
\end{tikzpicture}} 
\resizebox*{\textwidth}{!}{
  \begin{tikzpicture}
    \draw (-9,-3.75) -- (9,-3.75); 
    \foreach \x in
    {-9, -8, -7, -6, -5, -4, -3, -2, -1, 0, 1, 2, 3, 4, 5, 6, 7, 8, 9}
      \draw (\x cm,-3.75cm+3pt) -- (\x cm,-3.75cm-3pt); 
    \foreach \x in {-9,-3,3,9} 
      \draw[ultra thick] (\x cm,-3.75cm+5pt) -- (\x cm,-3.75cm-5pt); 
    \draw (-9,-3.75) node[below=5pt] {$1$}; 
    \draw (-3,-3.75) node[below=5pt] {$k_1$}; 
    \draw (3,-3.75) node[below=5pt] {$k_2$}; 
    \draw (9,-3.75) node[below=5pt] {$n$}; 
  \end{tikzpicture}
} \caption{ Illustration of the data generating model in Example~\ref{ex:lscm}
rolled out over time. The model remains fixed between $1$ and $k_1-1$, between
$k_1$ and $k_2-1$ and between $k_2$ and $n$. We intervene at two time points
$k_1$ and $k_2$, and the hammers indicate on which node these interventions act
with respect to the previous time interval. The population OLS coefficient
$\beta^{\operatorname{OLS}}_i$ changes at both time points $k_1$ and $k_2$ due
to the interventions (see details in Appendix~\ref{app:ex_details}). However,
only at $k_2$ the causal mechanism of $Y$ changes (at $k_1$ the causal mechanism
of $Y$ remains unchanged).}
\label{fig:ex_lscm}
\end{figure}
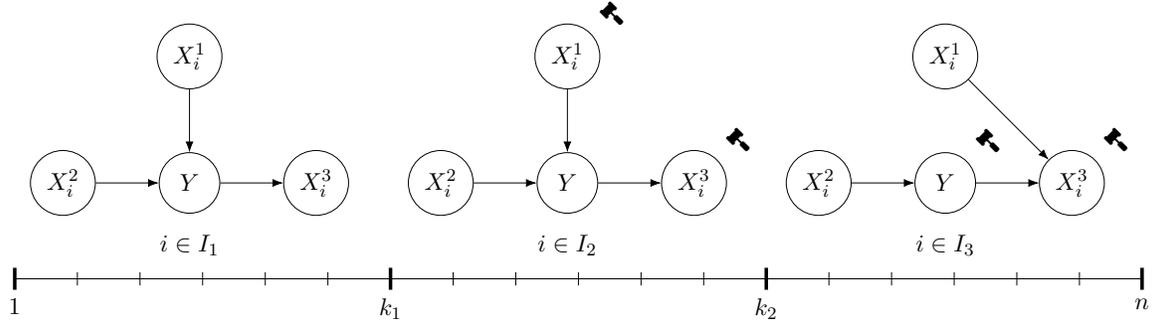

In this example, at both time points $k_1$ and $k_2$, the joint distribution of
$(X_i,Y_i)$ and in particular the population OLS parameter
$\beta^{\operatorname{OLS}}_i$ changes (i.e., $k_1$ and $k_2$ are both RCPs by
Definition~\ref{def:rcp}). However, the causal mechanism of the response $Y$
with respect to $X$ as specified in \eqref{eq:ex1causal} only changes at $k_2$.
Our proposed notion of causal change points defined in Definition~\ref{def:ccp}
below is able to capture this distinction. This example also highlights the
invariance property of causal models: Interventions that do not act directly on
the response $Y_i$ may change $Y_i\mid X_i$ but they keep $Y_i\mid
X_i^{\operatorname{PA}(Y_i)}$ invariant, where
$\operatorname{PA}(Y_i)\subseteq\{1,\ldots,4\}$ denotes the causal parents of
$Y_i$. A more formal treatment of causal models is provided in
Section~\ref{sec:causal_models}.
\end{example}

Example~\ref{ex:lscm} shows that it is possible that the conditional expectation
of $Y$ given $X$ can change even though the causal mechanism of how $Y$ is
affected by $X$ remains fixed. We propose to distinguish between changes only in
the full conditional expectation of $Y$ given $X$ and changes that manifest in
differences in the conditional expectations of $Y$ given $X^S$ for all
$S\in\cS$. Arguably, the second notion of change is of a more fundamental nature
and indicates a more drastic shift in the data generating process. To formalize
this notion which we call causal change points (see Definition~\ref{def:ccp}
below) we first define the population OLS coefficient and the corresponding
residuals based on subsets of covariates.

\begin{definition}[Population OLS given subsets of covariates]
\label{def:pop_ols} Assume that $\bbE[X_iX_i^{\top}]$ is invertible for all
$i\in\{1,\ldots,n\}$. For all $S\in\cS$ and all $i\in\{1,\ldots,n\}$, the
\emph{population OLS coefficient given $S$} is defined as
$\beta_i^{\text{OLS}}(S) \in \bbR^{d+1}$ satisfying
\begin{equation*}
\beta_i^{\text{OLS}}(S)^S = \bbE\left[X_i^S(X_i^S)^\top\right]^{-1}
\bbE\left[X_i^SY_i\right] 
\end{equation*} 
and $\beta_i^{\text{OLS}}(S)^j = 0$ for all $j\in\{1,\ldots,d+1\}\setminus S$.
The corresponding \emph{population OLS residual given $S$}  is defined as
$\epsilon_i(S) \coloneqq Y_i - X_i^\top\beta_i^{\operatorname{OLS}}(S)$. We use
the convention that
$\beta_i^{\text{OLS}}=\beta_i^{\text{OLS}}(\{1,\ldots,d+1\})$ and
$\epsilon_i=\epsilon_i(\{1,\ldots,d+1\})$. 
\end{definition}

Using this definition, we can now define what we call causal change points, the
time points at which for all subsets of covariates $S\in\cS$, either the
population OLS coefficient given $S$ or the distribution of the population OLS
residual given $S$ differs from previous time points (see
Definition~\ref{def:ccp}). Even though we call these changes ``causal'', the
definition does not rely on an underlying causal model. Nevertheless, the
definition of causal change points is motivated by the fact that under
additional causal assumptions, they correspond to changes in the causal
mechanism of $Y$ on $X$. We discuss this connection in
Section~\ref{sec:causal_models}.

\begin{definition}[Causal change point (CCP)] \label{def:ccp} A time point $k
\in \{2,\ldots,n\}$ is called a \emph{causal change point} (CCP) if for all
$S\in\cS$
\begin{equation*}
  \text{either}\quad\beta_k^{\operatorname{OLS}}(S)\neq\beta_{k-1}^{\operatorname{OLS}}(S)
  \quad\text{or}\quad
  \epsilon_k(S) \dneq \epsilon_{k-1}(S).
\end{equation*}
\end{definition}

By definition, CCPs form a subset of RCPs. We refer to RCPs that are not CCPs as
\emph{non-causal change points} (NCCPs). An alternative way to characterize CCPs
is via sets $S\in\cS$ for which the population OLS coefficient and residual
distribution given $S$ remain unchanged within a time interval.

\begin{definition}[Invariant set] \label{def:invsets} For a time interval
$I\in\cI$, a set $S\in\cS$ is called an $I$-\emph{invariant set} if there exists
a vector $\beta\in\mathbb{R}^{d+1}$ and a distribution $F$ such that for all
$i\in I$, 
\begin{equation*}
  \beta_i^{\operatorname{OLS}}(S)=\beta
  \quad\text{and}\quad
  \epsilon_i(S)\iid F.
\end{equation*}
\end{definition}

The following proposition characterizes CCPs in terms of invariant
sets. The proof, given in Appendix~\ref{app:proofs} for completeness,
follows directly from the definitions.

\begin{proposition}[Alternative characterization of CCP]
\label{prop:ccp_alt_char} A time point $k \in \{2,\ldots, n\}$ is a CCP if and
only if there does not exist a $\{k-1,k\}$-invariant set $S\in\cS$.
\end{proposition}

In the following section, we discuss how CCPs relate to changes of causal
mechanism when assuming an underlying causal model.

\subsection{Causal models as data generating models}\label{sec:causal_models}

We now formalize changes in causal mechanisms and relate them to CCPs. To this
end, we introduce a class of SCMs \citep[e.g.,][]{pearl2009causality} that
satisfy the assumptions of our sequential model, and discuss under which
additional causal assumptions, CCPs correspond to causal mechanism changes.
Furthermore, we argue in Example~\ref{ex:lscm2} that even if these assumptions
are not satisfied, CCPs capture meaningful changes.

\begin{setting}[Sequential linear SCM with hidden variables]
\label{set:seq_lscm_with_hidden} Let
$(X_1,Y_1),\ldots,(X_n,Y_n)\in\bbR^{d+1}\times\bbR$ be a sequence of observed
variables and $(H_1,\ldots,H_n)\in\bbR^{q}$ a sequence of unobserved variables.
For all $i\in\{1,\ldots,n\}$ consider an SCM over $(H_i, X_i, Y_i)$ given by
$X_i^{d+1} \coloneqq 1$ as intercept and
\begin{subequations}
\begin{align}
X_i^{S^\star} &\coloneqq A_iX_i + \alpha_i Y_i + h_i(H_i,\epsilon_i^X)  \label{eq:seq_lscm_x} \\
Y_i &\coloneqq \beta_i^\top X_i + g_i(H_i,\epsilon_i^Y), \label{eq:seq_lscm_y}
\end{align}
\end{subequations}
where $S^\star = \{1,\ldots,d\}$, $H_i$, $\epsilon_i^X$ and $\epsilon_i^Y$ are
jointly independent, $g_i$ and $h_i$ are arbitrary measurable functions such
that $\bbE[h_i(H_i,\epsilon_i^Y)]=\bbE[g_i(H_i,\epsilon_i^Y)]=0$. Furthermore,
the parameters in~\eqref{eq:seq_lscm_x} and~\eqref{eq:seq_lscm_y} are such that
for all $i \in \{1, \ldots, n\}$ the induced graph\footnote{For all time points
$i\in\{1,\ldots,n\}$ the graph is constructed by taking the observed variables
$X_i^1,\ldots, X_i^4, Y_i$ as nodes and adding a directed edge from node $V$ to
$W$ if variable $V$ appears with a non-zero coefficient in the structural
equation of variable $W$.} is directed and acyclic. For all
$i\in\{2,\ldots,n-1\}$ the set of (observed) parent variables of $Y_i$ is given
by $\text{PA}(Y_i) = \{j\in \{1,\ldots,d+1\}\mid \beta_i^j \neq 0\}$.
\end{setting}
 
Given such a causal model, we can characterize what CCPs correspond to under
certain conditions. In Proposition~\ref{prop:ccp_under_seq_lscm_with_hidden}, we
show that as long as the noise term of $Y$ remains uncorrelated with its
parents, a CCP indicates a change in either the causal coefficient $\beta_i$ or in
the noise term $g_i(H_i,\epsilon_i^Y)$. 

\begin{proposition}
\label{prop:ccp_under_seq_lscm_with_hidden} Assume
Setting~\ref{set:seq_lscm_with_hidden}, let $k\in\{2,\ldots,n\}$ be a fixed time
point and assume that for all $i \in \{k-1,k\}$ it holds that
\begin{equation}\label{eq:no_confound}
  \bbE[X_i^{\text{PA}(Y_{i})}g_i(H_i, \epsilon_i^Y)]=0.
\end{equation}
Then, it holds that
\begin{equation*}
k\text{ is a CCP}
\quad\Longrightarrow\quad
\beta_k \neq \beta_{k-1}\text{ or } g_k(H_k,\epsilon_k^Y) \dneq g_{k-1}(H_{k-1},\epsilon_{k-1}^Y).
\end{equation*}
\end{proposition}
A proof is given in Appendix~\ref{app:proofs}.  In the following example, we
illustrate that the statement of
Proposition~\ref{prop:ccp_under_seq_lscm_with_hidden} may be false if there is
hidden confounding in the sense that \eqref{eq:no_confound} is violated.

\begin{example}[CCPs with hidden confounding] \label{ex:lscm2} For all
$i\in\{1,\ldots,n\}$, consider the linear SCM over the variables
$(H_i,X_i^1,X_i^2,Y_i)$ given by $X_i^2\coloneqq 1$ as the intercept and
\begin{subequations}
\begin{align}
  X_i^1 &\coloneqq \alpha_i H_i + \epsilon_i^{X^1} \label{eq:ex_x} \\
  Y_i &\coloneqq X_i^1 + H_i + \epsilon_i^Y, \label{eq:ex_y}
\end{align}  
\end{subequations}
where $\epsilon_i^{X^1}, \epsilon_i^Y\iid\cN(0,1)$ for all $i\in\{1,\ldots,n\}$,
$H_i\iid\cN(0,1)$ for all $i\in\{1,\ldots,k_1-1\}$, $H_i\iid\cN(0,2)$ for all
$i\in\{k_1,\ldots,n\}$, $\alpha_i = 1$ for all $i\in\{1,\ldots,k_2-1\}$, and
$\alpha_i = 0$ for all $i\in\{k_2,\ldots,n\}$. For all $i\in\{1,\ldots,n\}$, the
variable $H_i$ is unobserved. Define the intervals $I_1 = \{1,\ldots,k_1-1\}$,
$I_2 = \{k_1,\ldots, k_2-1\}$, and $I_3 = \{k_2, \ldots,n\}$. The corresponding
DAGs are shown in Figure~\ref{fig:ex_lscm2}.

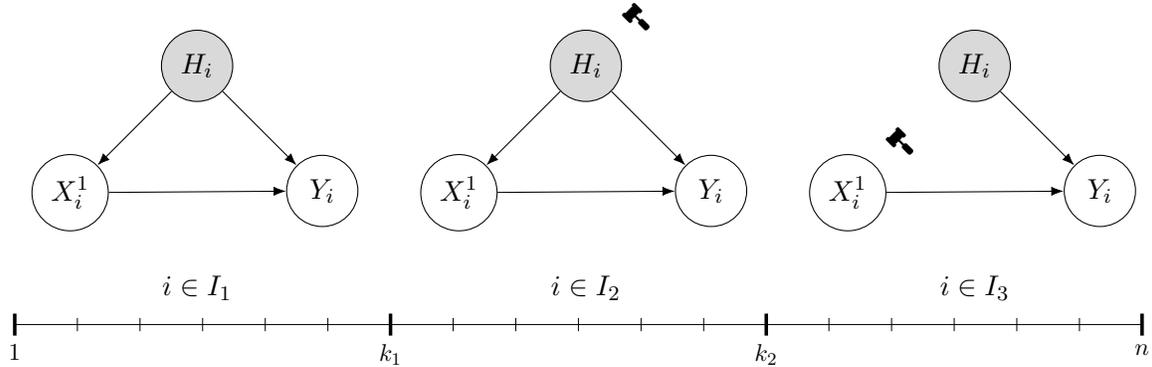
\begin{figure}[t]
\centering
\resizebox{0.95\textwidth}{!}{
\begin{tikzpicture}
  
  \node[state, fill=gray!30] at (0, 0) (H) {$H_i$};
  \node[below left=of H, state] (X) {$X_i^1$};
  \node[below right=of H, state] (Y) {$Y_i$};
  \draw[-Latex] (H) -- (X);
  \draw[-Latex] (H) -- (Y);
  \draw[-Latex] (X) -- (Y);
  \node[] at (0, -3){$i \in I_1$};
  
  \node[state, fill=gray!30] at (5.25, 0)(H) {$H_i$};
  \node[above right=0.05px of H] () {\faLegal};
  \node[below left=of H, state] (X) {$X_i^1$};
  \node[below right=of H, state] (Y) {$Y_i$};
  \draw[-Latex] (H) -- (X);
  \draw[-Latex] (H) -- (Y);
  \draw[-Latex] (X) -- (Y);
  \node[] at (5.25, -3){$i \in I_2$};
  
  \node[state, fill=gray!30] at (10.5, 0)(H) {$H_i$};
    \node[below left=of H, state] (X) {$X_i^1$};
    \node[below right=of H, state] (Y) {$Y_i$};
    \draw[-Latex] (H) -- (Y);
    \draw[-Latex] (X) -- (Y);
    \node[above right=0.05px of X] () {\faLegal};
    \node[] at (10.5, -3){$i \in I_3$};
\end{tikzpicture}}
\resizebox*{\textwidth}{!}{
\begin{tikzpicture}
  \draw (-9,-3.75) -- (9,-3.75); 
  \foreach \x in
  {-9, -8, -7, -6, -5, -4, -3, -2, -1, 0, 1, 2, 3, 4, 5, 6, 7, 8, 9}
    \draw (\x cm,-3.75cm+3pt) -- (\x cm,-3.75cm-3pt); 
  \foreach \x in {-9,-3,3,9} 
    \draw[ultra thick] (\x cm,-3.75cm+5pt) -- (\x cm,-3.75cm-5pt); 
  \draw (-9,-3.75) node[below=5pt] {$1$}; 
  \draw (-3,-3.75) node[below=5pt] {$k_1$}; 
  \draw (3,-3.75) node[below=5pt] {$k_2$}; 
  \draw (9,-3.75) node[below=5pt] {$n$}; 
\end{tikzpicture}
} \caption{Illustration of the data generating model in Example~\ref{ex:lscm2}
rolled out over time. The model remains fixed between $1$ and $k_1-1$, between
$k_1$ and $k_2-1$ and between $k_2$ and $n$. We intervene at two time points
$k_1$ and $k_2$, and the hammers indicate on which node these interventions act
with respect to the previous time interval. Even though both $k_1$ and $k_2$ are
CCPs, the causal mechanism of $Y$ with respect to $X$ only changes at $k_1$ (the
noise term $H_i + \epsilon_i^Y$ changes in distribution) but not at $k_2$
(neither the causal coefficient nor the noise term's distribution changes).} 
\label{fig:ex_lscm2}
\end{figure}

Here, both $k_1$ and $k_2$ are CCPs. To see this, consider the population OLS
parameter given $S_1 = \{1,2\}$ which is equal to $\beta_i^{\text{OLS}}(S_1) =
(c_i, 0)^\top$, where
\begin{equation*}
c_i = \frac{\cov(X_i^1, Y_i)}{\bbV(X_i^1)} = \begin{cases*}
  3/2 & $i\in I_1$ \\
  5/3 & $i\in I_2$ \\ 
  1 & $i\in I_3$.
\end{cases*}
\end{equation*}
Hence, for all $k\in\{k_1,k_2\}$ the set $S_1$ is not $\{k-1,k\}$-invariant
Moreover, the population OLS residual given $S_2 = \{2\}$ is given by
\begin{equation*}
\epsilon_i(S_2)
\sim \begin{cases*}
\cN(0, 6) & $i\in I_1$\\
\cN(0, 10) & $i\in I_2$\\
\cN(0, 4) & $i\in I_3$.
\end{cases*}
\end{equation*}
Again this implies that for all $k\in\{k_1, k_2\}$ the set $S_2$ is not
$\{k-1,k\}$-invariant. By Proposition~\ref{prop:ccp_alt_char}, both $k_1$ and
$k_2$ are CCPs. This shows that in the case of hidden confounding between $Y$
and $X^{\operatorname{PA}(Y)}$, it is no longer true that the existence of  a
CCP implies a change in the causal mechanism of $Y$ (as specified in
\eqref{eq:ex_y}) as in the unconfounded case considered in
Proposition~\ref{prop:ccp_under_seq_lscm_with_hidden}. Nevertheless, we consider
time points like $k_1$ and $k_2$ as conceptually different from other RCPs in
that they represent changes in what can be thought of as the observed causal
mechanism.
\end{example}

\section{Causal change point detection} \label{sec:ccpd} 

We now consider how to detect CCPs, that is, given a time interval $I \in\cI$,
we would like to decide whether there exists a CCP $k\in I$. For a fixed time
interval $I \in \cI$, the absence of CCPs in $I$ is equivalent (see
Proposition~\ref{prop:ccp_alt_char}) to the null hypothesis
\begin{equation} \label{hypo:no_ccp}
\cH_0^I: \exists S\in\cS \text{ s.t. } S \text{ is }
I\text{-invariant}.
\end{equation} 
The goal is to construct a (possibly randomized) hypothesis test $\varphi_I:
\bbR^{|I| \times (d+1)} \times \bbR^{|I|} \to \{0,1\}$ for $\cH_0^I$. The test
$\varphi_I$ is a function of the data $(\bX_I, \bY_I)$ and rejects the null
hypothesis $\cH_0^I$ if $\varphi_I(\bX_I, \bY_I)=1$ and does not reject it if
$\varphi_I(\bX_I, \bY_I)=0$. $\varphi_I$ is said to be level $\alpha \in (0,1)$
if $\sup_{P\in\cH_0^I} \bbP_P\left(\varphi_I(\bX_I, \bY_I) = 1\right) \leq
\alpha$, and it is said to have power $\beta\in (0,1)$ against an alternative
$P\not\in\cH_0^I$ if $\bbP_{P}\left(\varphi_I(\bX_I, \bY_I) = 1\right) = \beta$.

Testing $\cH_0^I$ can be split up into a multiple testing problem by considering
for all $S\in\cS$ the null hypothesis
\begin{equation}\label{hypo:s_inv}
  \cH_{0,S}^I: S \text{ is } I\text{-invariant}.
\end{equation}
This null hypothesis equals the hypothesis that the population OLS coefficient
and residuals given $S$ do not change. For such settings, tests have been
derived previously (see Section~\ref{sec:explicit_tests}). Given a collection of
tests $(\varphi_I^S)_{S\in\cS}$ for the null hypotheses
$(\cH_{0,S}^I)_{S\in\cS}$ that are at level $\alpha$, we can combine them to a
test for $\cH_0^I$ as follows.

\begin{proposition} \label{prop:correct_level} Let $(\varphi_I^S)_{S\in\cS}$ be
a family of tests for the hypotheses $(\cH_{0,S}^I)_{S\in\cS}$ where for all
$S\in\cS$, $\varphi_I^S: \bbR^{|I|\times |S|} \times \bbR^{|I|}\to\{0, 1\}$, and
$\varphi_I^S$ is level $\alpha \in (0, 1)$. Then the test $\varphi_I:\bbR^{|I|
\times (d+1)} \times \bbR^{|I|} \to \{0,1\}$ defined for all $x\in\bbR^{|I|
\times (d+1)}$ and all $y\in\bbR^{|I|}$ by
\begin{equation*} 
\varphi_I(x, y) \coloneqq 
\begin{cases}
1 & \text{if } \min_{S\in\cS}\varphi_I^S(x^S, y) = 1 \\
0 & \text{otherwise}
\end{cases}
\end{equation*}
is level $\alpha$ for $\cH_0^I$.
\end{proposition}
The proof of this proposition is straightforward and is included in
Appendix~\ref{app:proofs}.

\subsection[Tests for H0S]{Tests for $\cH_{0,S}^I$} \label{sec:explicit_tests}

We first introduce the following definition of population OLS parameter and
residuals over a time interval given a subset of covariates.
\begin{definition}[Population OLS over an interval given a subset of covariates]
\label{def:pop_ols_interval} Let $I \in \cI$ and assume $\sum_{i\in
I}\bbE[X_iX_i^\top]$ is invertible. For all $S\in\cS$, the \emph{population OLS
parameter given $S$ over $I$} is the vector
$\beta_I^{\text{OLS}}(S)\in\bbR^{d+1}$ where 
\begin{equation*}
\Big(\beta_I^{\text{OLS}}(S)\Big)^S = \Big[\sum_{i\in I}\bbE[X_i^S(X_i^S)^\top]\Big]^{-1}\sum_{i\in I}\bbE(X_i^SY_i)
\end{equation*}
and $\big(\beta_I^{\text{OLS}}(S)\big)^j = 0$ for all
$j\in\{1,\ldots,d+1\}\setminus S$. For all $\ell\in \{1, \ldots,n\}$, the
\emph{population OLS residual at $\ell$ given $S$ over $I$} is given by
$\epsilon_\ell^{I}(S) \coloneqq Y_\ell - X_\ell^\top\beta_{I}^{\text{OLS}}(S)$
with the convention that $\epsilon^{I}(S) = \Big(Y_\ell -
X_\ell^\top\beta_{I}^{\text{OLS}}(S)\Big)_{\ell\in I}\in\bbR^{|I|}$. 
\end{definition}
One way to test the null hypothesis $\cH_{0,S}^I$ is to first divide the
interval $I$ into two sub-intervals $I_1 \coloneqq \{\min(I),\ldots, \min(I) +
\lfloor \frac{|I|}{2}\rfloor - 1\}$ and $I_2\coloneqq\{\min(I) +
\lfloor\frac{|I|}{2}\rfloor,\ldots,\max(I)\}$. Then, $\cH_{0,S}^I$ in
\eqref{hypo:no_ccp} implies 
\begin{equation} \label{hypo:two_env}
\cH_{0,S}^{I_1, I_2}: \beta_{I_1}^{\text{OLS}}(S) = \beta_{I_2}^{\text{OLS}}(S) \text{ and } \epsilon^{I_1}(S) \deq \epsilon^{I_2}(S).
\end{equation}
The reverse implication, however, is not true in general: $\cH_{0,S}^{I_1, I_2}$
does not generally imply $\cH_{0,S}^I$. \eqref{hypo:two_env} can be tested by
e.g., the Chow test \citep{chow1960tests}. Details of the Chow test are given in
Appendix~\ref{app:chow_test}. A version of this test has also been suggested for
Invariant Causal Prediction \citep{peters2016causal}. As an alternative, one can
use the procedure proposed by \citet{pfister2019invariant}: instead of two
sub-intervals, one considers a pre-defined grid over the time indices and
combines test statistics by resampling scaled versions of the residuals. 

\section{Causal change point localization} \label{sec:ccpl}

We now discuss two approaches for estimating the locations of CCPs. The first
approach is based on testing candidates. By Definition~\ref{def:ccp}, CCPs are a
subset of RCPs. Thus, if we are given the set of RCPs, we can use the detection
method described in Section~\ref{sec:ccpd} to identify the CCPs among them. An
alternative approach is based on a loss function that aims to detect the CCPs
directly. For localizing multiple CCPs, we can combine the proposed loss
function (see Definition~\ref{def:causal_stablility_loss}) with existing
multiple change point localization algorithms. Popular multiple change point
localization algorithms are often of two types: algorithms based on dynamic
programming \citep[e.g.,][]{hawkins1976point, killick2012optimal} and greedy
algorithms \citep[e.g.,][]{vostrikova1981detecting, fryzlewicz2014wild}. In
order to estimate the locations of all CCPs, both approaches rely on statistical
methods for detecting changes in both regression parameters and the residual
distributions. 

Throughout this section, we assume there exists a set of $q \in \{1,\ldots,
n-2\}$ CCPs $\cT \coloneqq\{\tau_1,\ldots,\tau_q\}$, where $\tau_i < \tau_{i+1}$
for all $i\in\{1,\ldots,q-1\}$ and we use the convention that $\tau_0\coloneqq
1$ and $\tau_{q+1}\coloneqq n+1$.

\subsection{Causal change point localization by pruning candidates}
\label{sec:ccpl_candi}

Assume we are given a candidate set
$\cK=\{k_1,\ldots,k_l\}\subseteq\{2,\ldots,n-1\}$ of potential CCPs. This could,
for example, be the set of RCPs or a superset of the RCPs (see
Definition~\ref{def:rcp}). For the purpose of this section, we assume that the
true CCPs are contained in $\cK$ but in practice one would estimate the set
$\cK$ using existing methods for localizing RCPs
\citep[e.g.,][]{bai1997estimating}, which may lead to violations of this
assumption. We can then prune the candidate set $\cK$ by testing whether a
candidate $k_j$ is indeed a causal change point considering the interval
$I_j=\{k_{j-1},\ldots,k_{j+1}-1\}$ (with the convention that $k_0 = 1$ and
$k_{l+1} = n+1$) and using a test for $\cH_0^{I_j}$ discussed in
Section~\ref{sec:ccpd}. The detailed procedure is provided in
Algorithm~\ref{alg:prune_candi} in Appendix~\ref{app:alg}. 

Proposition~\ref{prop:candi_prop} gives lower bounds (which are functions of the
properties of the test) of the probability that Algorithm~\ref{alg:prune_candi}
localizes only the true CCPs and the probability that
Algorithm~\ref{alg:prune_candi} localizes all the true CCPs when the candidate
set is a superset of the true CCPs.

\begin{proposition}\label{prop:candi_prop} Denote by
$\cT\subseteq\{2,\ldots,n\}$ the set of CCPs and by
$\cK=\{k_1,\ldots,k_L\}\subseteq\{2,\ldots,n\}$, for $L\geq 1$, a candidate set
of CCPs satisfying $\cT\subseteq\cK$. Moreover, denote for all
$\ell\in\{1,\ldots,L\}$, the intervals
$I_{k_\ell}\coloneqq\{k_{\ell-1},\ldots,k_{\ell+1}-1\}$, where $k_0=1$ and
$k_{L+1}=n+1$. Let $\widehat{\cT}$ be the CCP estimator defined in
Algorithm~\ref{alg:prune_candi} and let $(\varphi_I)_{I\in\cI}$ be a collection
of tests for $(\cH_0^I)_{I\in\cI}$.  Then, the following two statements hold:
\begin{itemize}
\item[(i)] Let $\alpha\in (0,1)$. If for all $k\in\cK$ it holds that
$\varphi_{I_k}$ is level $\alpha$, then
\begin{equation*}
  \mathbb{P}(\widehat{\cT}\subseteq \cT)\geq 1- (|\cK|-|\cT|)\cdot\alpha.
\end{equation*}
\item[(ii)] Let $\beta\in (0,1)$. If for all $\ell\in\{1,\dots,L\}$ with
$k_{\ell}\in\cT$ it holds that $\mathbb{P}(\varphi_{I_{k_{\ell}}}=1)\geq\beta$,
then
\begin{equation*}
  \mathbb{P}(\cT\subseteq \widehat{\cT})\geq 1- |\cT|\cdot(1-\beta).
\end{equation*}  
\end{itemize}
\end{proposition}

A proof can be found in Appendix~\ref{app:proofs}. Following
Proposition~\ref{prop:candi_prop}, one may adjust $\alpha$ by a factor $c\leq
1/(|\cK| - |\cT|)$ which ensures that $\bbP(\widehat{\cT} \subseteq\cT)\geq
1-\alpha$. One special case is the Bonferroni correction, which corresponds to
$c = 1/|\cK|$ and always preserves coverage at level $\alpha$ but might be
conservative if there are many CCPs. In practice, the candidate set may not be a
superset of the true CCPs, i.e., the candidate set may not contain all the true
CCPs, or some candidates are time points that slightly deviate from the true
CCPs, or a combination of both. If this is the case, the resulting CCP estimates
can be arbitrarily biased. To check the validity of the estimates, one can test
each of the sub-intervals separated by $\widehat\cT$: if there is no invariant
set in a sub-interval, it means that there exist at least one CCP that is not in
the candidate set, or the candidates surrounding the sub-interval are not true
CCPs. We illustrate this in Appendix~\ref{app:add_exp}.

\subsection{Causal change point localization via a loss function}
\label{sec:ccpl_loss}

An alternative approach to localizing causal change points is by finding the
minima of a loss function. Here, we propose a loss function that assesses the
level of causal non-invariance at each time point. Ideally, the loss function
should achieve its minimal value in an interval with a single CCP at the true
CCP. In Section~\ref{sec:inv_loss_pop}, we introduce the loss function and
discuss its properties at population level given a single CCP. We discuss how to
estimate the location of one CCP using an empirical version of the loss function
in Section~\ref{sec:inv_loss_sample}. Localization of multiple CCPs is discussed
in Section~\ref{sec:multi_ccps}, where we leverage modified versions of existing
multiple change point detection algorithms \citep{vostrikova1981detecting,
baranowski2019not, kovacs2020seeded} to localize each of them.

\subsubsection{Causal stability loss at population level}
\label{sec:inv_loss_pop}

In this section, we introduce a loss function to capture the change in causal
mechanism of the response $Y$. Intuitively, for an interval $I\in\cI$, the loss
at a time point $i\in I$ sums up the level of non-invariance over the two
sub-intervals of $I$ to the left and right of $i$. Suppose there exists exactly
one CCP in $I$, this loss function achieves its minimum value at the true CCP at
population level. To formally introduce the loss function, we require the
following notation.

\begin{notation2}
Let $s\in\bbN$ be a minimal segmentation length. For all intervals
$I\in\mathcal{I}$ with $|I| \geq 2s$, let $m_s(I)\coloneqq
\lfloor\frac{|I|}{s}\rfloor$ and let $P_1(I),\ldots,P_{m_s(I)}(I)$ be a
partition of $I$ into $m_s(I)$ intervals such that $|P_r(I)| = s$ for $r
\in\{1,\ldots, m_s(I)-1\}$ and $|P_{m_s(I)}(I)| = |I| - \big(m_s(I)-1\big)\cdot
s$. For all $r\in\{1,\ldots, m_s(I)\}$, we denote the complement of $P_r(I)$ as
$P_r^c(I) \coloneqq I\setminus P_r(I)$. For all $I\in\cI$ with $|I| < 2s$, we
let $m_s(I) = 1$, $P_1(I) = I$, and with a slight abuse of notation, $P_1^c(I) =
I$. An illustration of this notation is given in Figure~\ref{fig:split}. For all
$I,J \in\cI$, we define $\displaystyle V_{I,J}(S) = \frac{1}{|I|}\sum_{\ell\in
I}\bbE\Big[\big(\epsilon_\ell^J(S)\big)^2\Big]$ where $\epsilon_\ell^J(S)$ is
the population OLS residual at $\ell$ given $S$ over $J$ (see
Definition~\ref{def:pop_ols_interval}).
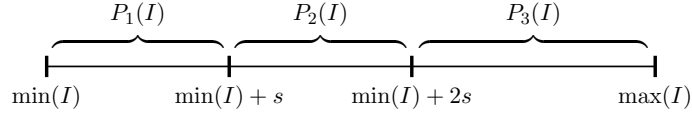
\begin{figure}[!hbt]
\centering
\input{tikz_split_new.tex}
\caption{Illustration of the partition of an interval $I$ with $3s \leq |I| \leq
4s-1$ into $3$ sub-intervals. The ticks $\min(I)$, $\min(I) + s$ and $\min(I) +
2s$ mark the beginning of the three intervals and $\max(I)$ marks the end of the
third interval.}
\label{fig:split}
\end{figure}
\end{notation2}

If there is no CCP in $I\in\cI$, then there exists $S\in\cS$ and $c \in
\mathbb{R}$ such that for all $J\subseteq I$, $V_{I\setminus J,J}(S) = V_{J,
J}(S) = c$. This motivates the following definition of minimal OLS instability
(Definition~\ref{def:min_stab_loss}) which serves as the basis of our causal
stability loss (Definition~\ref{def:causal_stablility_loss}).

\begin{definition}[Minimal OLS instability]\label{def:min_stab_loss} Let
$I\in\cI$ and let $s\in\bbN$ be a minimal segmentation length. The \emph{minimal
OLS instability} over the interval $I$ is defined as 
\begin{align*}
\cC_{s}(I) \coloneqq \min_{S\in\cS} \sum_{r=1}^{m_s(I)} \Big(
  V_{P_r^c(I), P_r(I)}(S) - V_{P_r(I), P_r(I)}(S)
\Big)^2.
\end{align*}
\end{definition}

It satisfies the following property. 

\begin{proposition}\label{prop:min_ols_ins_wo_ccp} Let $I\in\cI$ and $s\in\bbN$.
Suppose there is no CCP in $I$, then $\cC_s(I) = 0$. 
\end{proposition}
A proof is given in Appendix~\ref{app:proofs}. We then define the causal
stability loss at a time point $i$ in an interval $I$ as the sum of minimal OLS
instability over the sub-intervals to the left and right of the time point $i$.

\begin{definition}[Causal stability loss]\label{def:causal_stablility_loss} Let
$I\in\cI$ and let $s\in\bbN$ be a minimal segmentation length. For all $i\in I
\setminus \{\min(I),\max(I)\}$, we define $I_{i-} \coloneqq
\{\min(I),\ldots,i-1\}$ and $I_{i+} \coloneqq \{i,\ldots,\max(I)\}$, then we
define the \emph{causal stability} loss as 
\begin{align*}
\cL_{I,s}(i) = \frac{\cC_s(I_{i-}) + \cC_s(I_{i+})}{m_s(I_{i-}) + m_s(I_{i+})},
\end{align*}
where $\cC_s(I_{i-})$ and $\cC_s(I_{i+})$ are as defined in
Definition~\ref{def:min_stab_loss}.
\end{definition}

The following property of the causal stability loss follows directly from
Proposition~\ref{prop:min_ols_ins_wo_ccp}.
\begin{corollary}
Let $I\in\cI$ and $s\in\mathbb{N}$. If there is no CCP in $I$, then for all
$i\in I$ $\cL_{I,s}(i) = 0$; if $\tau \in \{\min(I)+1,\ldots,\max(I)-1\}$ is the
only CCP in $I$, then $\cL_{I,s}(\tau) = 0$.
\end{corollary}

\subsubsection{Localizing a single causal change point}
\label{sec:inv_loss_sample}

The loss function $\cL_{I,s}$ can be estimated by replacing the population
quantities with their empirical counterparts. The OLS coefficient given
$S\in\cS$ over an interval $I \in \cI$ can be estimated by the vector
$\hat{\beta}^{\text{OLS}}_I(S)\in\bbR^{d+1}$ where 
\begin{equation*}
\left(\hat{\beta}^{\text{OLS}}_I(S)\right)^S = \argmin_{\beta^S\in\mathbb{R}^{|S|}}\sum_{\ell\in I}\left(Y_\ell - (X_\ell^S)^\top\beta^S\right)^2
\end{equation*}
and $\big(\hat{\beta}_I^{\text{OLS}}(S)\big)^j = 0$ for all
$j\in\{1,\ldots,d+1\}\setminus S$. For all $I\in\cI$, let
$\hat\beta_I^\text{OLS}(S)$ be the estimated OLS coefficient given $S$ over $I$.
For all $\ell\in\{1,\ldots,n\}$, the estimated OLS residual at $\ell$ given $S$
over $I$ is given by $\hat\epsilon_\ell(S) = Y_\ell -
X_\ell^\top\hat{\beta}^{\text{OLS}}_I(S)$ with the convention that
$\hat\epsilon^{I}(S) = \Big(Y_\ell -
X_\ell^\top\hat\beta_{I}^{\text{OLS}}(S)\Big)_{\ell\in I}\in\bbR^{|I|}$. Lastly,
for all $I, J\in\cI$, let $\widehat{V}_{I, J}^2(S) \coloneqq
\frac{1}{|I|}\sum_{\ell\in I} \left(\hat\epsilon_\ell^J\left(S\right)\right)^2$.
Then the minimal OLS instability over $I$ can be estimated by 
\begin{equation*}
\widehat\cC_s(I) \coloneqq
\min_{S\in\cS}\sum_{r=1}^{m(I)}\Big(\widehat{V}_{P_r^c(I), P_r(I)}(S) -
\widehat{V}_{P_r(I), P_r(I)}(S)\Big)^2
\end{equation*}
and the causal stability loss can be estimated by 
\begin{equation*}
\widehat\cL_{I,s}(i) \coloneqq \frac{\widehat\cC_s(I_{i-}) + \widehat\cC_s(I_{i+})}{m(I_{i-}) + m(I_{i+})}.
\end{equation*}

\subsubsection{Localizing multiple causal change points} \label{sec:multi_ccps}

We propose two general approaches to localize multiple CCPs. The first approach
uses the standard binary segmentation algorithm proposed by
\citet{vostrikova1981detecting} (see Algorithm~\ref{alg:bin_seg} in
Appendix~\ref{app:alg}) and then prunes the resulting estimates by
Algorithm~\ref{alg:prune_candi} in Appendix~\ref{app:alg}. The pruning step is
necessary since even at population level, the causal stability loss does not
necessarily achieve its minimum at one of the true CCPs in an interval that
contains multiple CCPs, although it does so when only one CCP exists in an
interval.

As a second approach, we consider a different greedy approach which instead of
searching for change points in a top-down order as the standard binary
segmentation, it searches for change points in a bottom-up order, namely the
seeded binary segmentation algorithm \citep{kovacs2020seeded} with the
narrowest-over-threshold selection procedure \citep{baranowski2019not}. The idea
is to first generate a collection of sets of intervals with increasing
lengths\footnote{The seeded intervals can be seen as generated layer by layer,
as described in Algorithm~\ref{alg:seeded_intervals}. The intervals on the same
layer are  considered to have the same length, despite the small differences
caused by rounding.}. Among the narrowest intervals for which $\cH_0^I$ is
rejected, we estimate one CCP in the interval that has the smallest p-value, and
eliminate all intervals that contain the estimated CCP. We then repeat the
procedure among the remaining sets of intervals from the narrowest to the widest
until $\cH_0^I$ is not rejected for any remaining intervals. The bottom-up order
aims to ensure that each interval only contains at most one CCP, which is
suitable when the loss function may not achieve its minimum at a CCP given
multiple CCPs in an interval. The procedure of obtaining the seeded intervals
\citep[Definition 1,][]{kovacs2020seeded} is given in
Algorithm~\ref{alg:seeded_intervals} in Appendix~\ref{app:alg}.
Algorithm~\ref{alg:seededBS} in Appendix~\ref{app:alg} describes the overall
procedure of the second approach. We compare the above approaches in
Section~\ref{sec:simulation}. 

\section{Numerical Experiments} \label{sec:experiments} 

We demonstrate the performance of our proposed methods based on simulated
datasets. In Section~\ref{sec:simulation} we describe the data generating
process that the simulated experiments are based on. For CCP detection, we show
the level and power given different numbers of splits with the Chow test for
testing $\cH_{0,S}^I$ as described in Section~\ref{sec:explicit_tests}. For CCP
localization, we consider both the case where it is known that there is exactly
one CCP and the case where there are multiple CCPs, and compare the methods
proposed in Section~\ref{sec:ccpl}. All numerical experiments can be reproduced
using the code available at \url{https://github.com/shimenghuang/CausalCP}.

\subsection{Simulated experiments} \label{sec:simulation}

In this section, we consider the following data generating process given for all
$i\in\{1,\ldots,n\}$ by
\begin{equation} \label{eq:five_node_mb_graph}
  \begin{split}
    X_i^1 &\coloneqq \epsilon_i^1 \\
    X_i^2 &\coloneqq \alpha_i^{12} X_i^1 + \epsilon_i^2 \\
    Y_i &\coloneqq \beta_i^{15} X_i^1 + \beta_i^{25} X_i^2 + \epsilon_i^Y \\
    X_i^4 &\coloneqq \epsilon_i^4 \\
    X_i^3 &\coloneqq \alpha_i^{53} Y_i + \alpha_i^{43} X_i^4 + \epsilon_i^3,
  \end{split}
\end{equation}
where $\epsilon_i^j \sim\cN(\mu_i^j, (\sigma_i^j)^2)$ for $j\in\{1,2,3,4,Y\}$.
The induced DAG for all $i\in\{1,\ldots,n\}$ is shown in
Figure~\ref{fig:five_node_mb_graph}. There are in total $15$ parameters in this
data generating process which can be divided into two sets: $4$ parameters that
are related to the causal mechanism of $Y$ with respect to $X$ ($\beta_i^{15}$,
$\beta_i^{25}$, $\mu_i^Y$, and $\sigma_i^Y$), and $11$ parameters that are not
related to causal mechanism of $Y$ with respect to $X$ ($\alpha_i^{12}$,
$\alpha_i^{53}$, $\alpha_i^{43}$, $\mu_i^j$ for $j\in\{1,2,3,4\}$, and
$\sigma_i^j$ for $j\in\{1,2,3,4\}$). Throughout this section, we refer to the
set of parameters that are related to the causal mechanism of $Y$ as the
\emph{causal parameters}, and the set of parameters thar are not related to
causal mechanism of $Y$ as \emph{non-causal parameters}. In the following
experiments, parameters are chosen such that changes in the causal parameters
are CCPs (see Definition~\ref{def:ccp}) and changes in the non-causal parameters
are NCCPs (we have verified this using straight-forward computations). More
details of the data generating process can be found in the
Appendix~\ref{app:exp_det}. All experiments are based on $200$ repetitions.

\begin{figure}[ht]
  \centering
  \begin{tikzpicture}
    \node[state] at (0, 0) (Y) {$Y_i$};
    \node[state] at (-1.5, 1.5) (X1) {$X_i^1$};
    \node[state] at (-1.5, -1.5) (X2) {$X_i^2$};
    \node[state] at (2, 0) (X3) {$X_i^3$};
    \node[state] at (4, 0) (X4) {$X_i^4$};
    \draw[-Latex] (X1) -- (X2);
    \draw[-Latex] (X1) -- (Y);
    \draw[-Latex] (X2) -- (Y);
    \draw[-Latex] (Y) -- (X3);
    \draw[-Latex] (X4) -- (X3);
  \end{tikzpicture}
\caption{DAG induced by \eqref{eq:five_node_mb_graph} for all
$i\in\{1,\ldots,n\}$.}
\label{fig:five_node_mb_graph}
\end{figure}
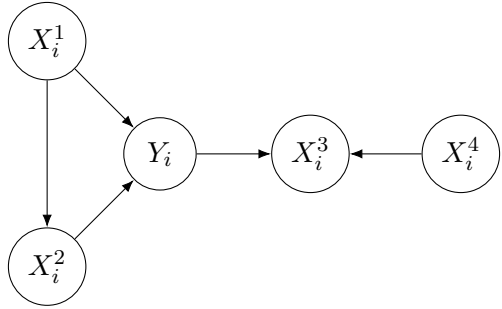

\subsubsection{CCP detection} \label{sec:exp_ccpd}

We demonstrate the power to detect CCPs where changes happen in the causal
parameters $\beta_i^{15}$ and $\beta_i^{25}$ in \eqref{eq:five_node_mb_graph} of
the procedure described in Section~\ref{sec:explicit_tests} based on the Chow
test and show that it holds the correct level. We fix $\alpha$ to be $0.05$ in
this experiment. Figure~\ref{fig:level_power} shows that the procedure has the
most power when the true CCP is in the middle of the interval while it has least
power when the true relative location of one single CCP is close to the
boundaries of the interval. A possible explanation is that the Chow test is
applied on the two sub-intervals to the left and right of the midpoint. When the
true CCP is to the left of the midpoint, the left half of the interval contains
data from a mixture of two distributions before and after the change, and the
right half contains data only from the distribution after the change, similarly
when the true CCP is to the right of the midpoint. Only when the true CCP
coincides with the midpoint, the two sides both contain data from a single
distribution which leads to a higher power with the Chow test. The ``no CCP''
label on the x-axis corresponds to when the interval does not contain a CCP. In
that case a valid method should control the type-I error of wrongly detecting a
CCP at the pre-specified level ($5\%$ in this case). 

\begin{figure}[ht]
\centering
\includegraphics[width=0.8\textwidth]{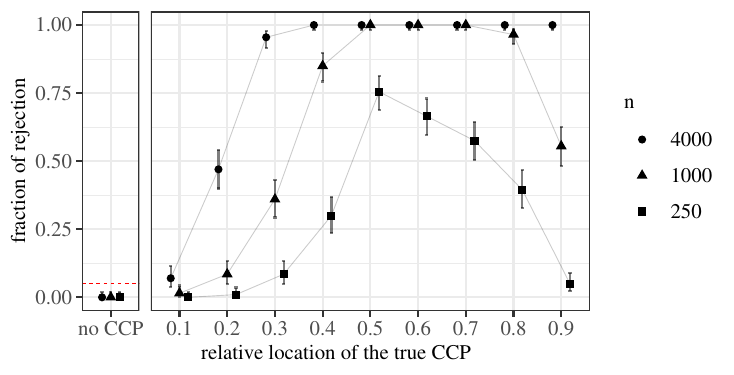}
\caption{Empirical investigation of level and power of the testing procedure
described in Section~\ref{sec:explicit_tests} with increasing sample sizes. The
x-axis corresponds to the relative location of a single CCP (and no CCP). The
error bars are binomial confidence intervals and the red dashed line is at
$0.05$.}
\label{fig:level_power}
\end{figure}

\subsubsection{CCP localization} \label{sec:exp_ccpl} 

We compare the proposed methods for localizing CCPs described in
Section~\ref{sec:ccpl} and the method breakpoints for localizing structural
changes in linear models due to \citet{bai2003computation} implemented in the R
package `strucchange' \citep{zeileis2003testing}. In Experiment 1 and Experiment
2, which consider datasets with a single true CCP, we apply the pruning approach
described in Section~\ref{sec:ccpl_candi} assuming the set of true RCPs is known
(\texttt{Prune-Oracle}) and assuming that they are unknown then using
breakpoints to estimate the RCPs (\texttt{Prune-BP}). These approaches are
compared with the causal stability loss (\texttt{LossCS}) approach described in
Section~\ref{sec:ccpl_loss} and using breakpoints to estimate a single change
point (\texttt{BP-1}). In Experiment 3 where there are two true CCPs, we apply
the following methods: (i) using breakpoints to estimate two change points
(\texttt{BP-2}), (ii) using breakpoints to estimate all change points up to a
specified minimal segment length (\texttt{BP}), (iii) using seeded binary
segmentation with the causal stability loss as in Algorithm~\ref{alg:seededBS}
in Appendix~\ref{app:alg} (\texttt{LossCS-SeededBS}), (iv) using
\texttt{LossCS-SeededBS} combined with a pruning step
(\texttt{LossCS-SeededBS-Prune}), (v) using standard binary segmentation with
the causal stability loss as in Algorithm~\ref{alg:bin_seg} in
Appendix~\ref{app:alg} (\texttt{LossCS-StdBS}), and (vi) using
\texttt{LossCS-StdBS} combined with a pruning step
(\texttt{LossCS-StdBS-Prune}).

\paragraph{Experiment 1: Fixed relative locations of one CCP and two NCCPs with
increasing $n$} For a total number of time points $n$, one true CCP is fixed at
$0.5 n + 1$ and two true NCCPs are fixed at $\lceil0.25 n + 1\rceil$ and
$\lceil0.75 n + 1\rceil$, respectively. The minimal segmentation length for all
methods is set to be $0.1n$. For \texttt{LossCS}, we evaluate the loss at every
$0.05 n$ time points starting from $\lceil0.05 n + 1\rceil$ and ending at
$\lceil0.95 n + 1\rceil$. We compare the different methods for sample sizes $n
\in\{250, 1000, 4000\}$. In Figure~\ref{fig:exp_set1}, we see that under this
setup, all methods except \texttt{BP-1} give estimates concentrating around the
true value with increasing sample size, while estimates based on \texttt{BP-1}
concentrate around the second NCCP. Moreover, \texttt{LossCS} performs better
than the two pruning approaches \texttt{Prune-Oracle} and \texttt{Prune-BP}
which both end up not detecting any CCPs in many simulations. Lastly, the
results of \texttt{Prune-Oracle} and \texttt{Prune-BP} are similar for large
sample sizes, indicating that estimating RCPs from data and using these
estimates as candidates can indeed perform well.

\begin{figure}[ht]
\centering
\includegraphics[width=0.99\textwidth]{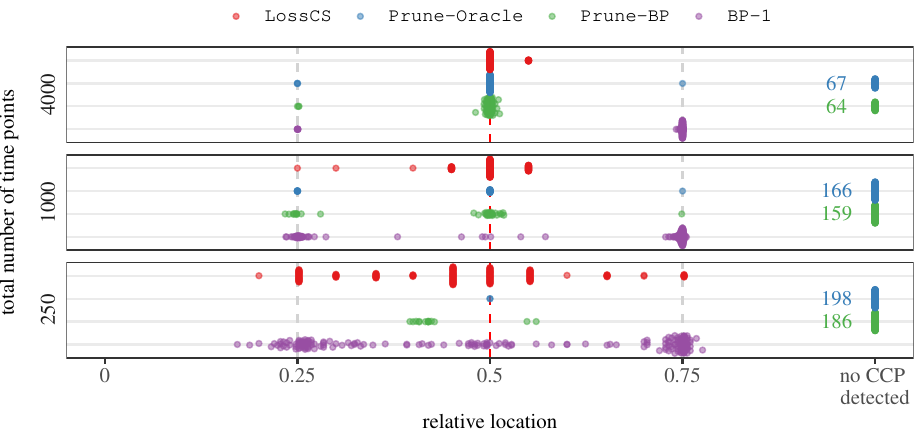}
\caption{Relative locations of the estimated CCP using different methods with
varying number of time points $n$. The red vertical dashed line indicates the
true relative location of the CCP and the two grey vertical dashed lines
correspond to the true relative locations of the NCCPs. With increasing sample
size, estimates from all methods except \texttt{BP-1} concentrate around the
true CCP, but \texttt{Prune-Oracle} and \texttt{Prune-BP} both detect no CCP in
many of the $200$ repetitions.}
\label{fig:exp_set1}
\end{figure}

\paragraph{Experiment 2: Different relative locations of a single CCP} For $n =
2000$, two true NCCPs are fixed at $0.25 n + 1$ and $0.75 n + 1$, respectively.
We evaluate the performance of \texttt{LossCS} when the relative location of a
single CCP is fixed at $\nu n + 1$ for $\nu\in\{0.1,0.2,\ldots,0.9\}$. We fix
the minimum segmentation length to be $0.1 n$ and for \texttt{LossCS} we
evaluate the loss at every $0.1n$ points. Figure~\ref{fig:exp_set2} shows the
estimated versus true relative locations of the CCP. The points are jittered
horizontally for visual clarity. We can see that \texttt{LossCS} performs well
when the true CCP is relatively close to the center of the interval, but when
the true CCP is close to the left (respectively, right) boundary of the
interval, \texttt{LossCS} tends to over- (respectively, under-) estimate the
location.

\begin{figure}[ht]
\centering
\includegraphics[width=0.55\textwidth]{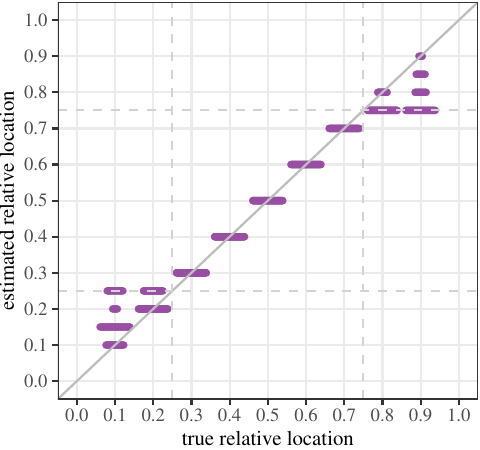}
\caption{Estimated relative locations the CCP given different true relative
locations of the CCP and $n=2000$. The two vertical dashed lines indicate the
(fixed) true relative locations of the NCCPs. \texttt{LossCS} tends to over-
(respectively, under-) estimate the location of CCP when the CCP is close to the
left (respectively, right) boundary.}
\label{fig:exp_set2}
\end{figure}

\paragraph{Experiment 3: Fixed relative locations of two CCPs and one NCCP with increasing $n$}

For a total number of time points $n$, two CCPs are fixed at $0.2 n + 1$ and
$0.8 n +1$, and one NCCP is fixed at $0.5 n+1$. We fix the minimum segmentation
length for all methods to be $0.2 n$ and for \texttt{LossCS} we evaluate the
loss at every time point other than the $\max(\lceil 0.1|I| \rceil, 10)$ points
at the beginning and end of each seeded interval $I$. We compare the different
methods for sample sizes $n \in\{1000, 2000, 4000\}$. Bonferroni correction is
applied in the pruning step when there is more than one candidate.
Figure~\ref{fig:exp_set3} contains the histogram of the estimates based on each
approach under this setting. As can be seen, using \texttt{BP-2} to estimate two
CCPs is not a valid method as it might detect NCCPs as in this example. With a
large enough sample size, \texttt{Prune-BP} performs best based on the number of
false and true positives, followed by \texttt{LossCS-SeedBS-Prune},
\texttt{LossCS-SeedBS}, and \texttt{LossCS-StdBS-Prune}. When the sample size is
small, \texttt{LossCS-SeedBS-Prune} performs best in the sense that it has the
least number of false positives and the most number of true positives. The
pruning step after \texttt{LossCS-SeedBS} does not seem to improve the
estimation much especially when the sample size is large, as the number of false
positives is already low. However, for \texttt{LossCS-StdBS} the pruning step is
crucial in this example, as it tends to first split on the NCCP hence leading to
many false positives.

\begin{figure}[!htb]
\centering
\includegraphics[width=\textwidth]{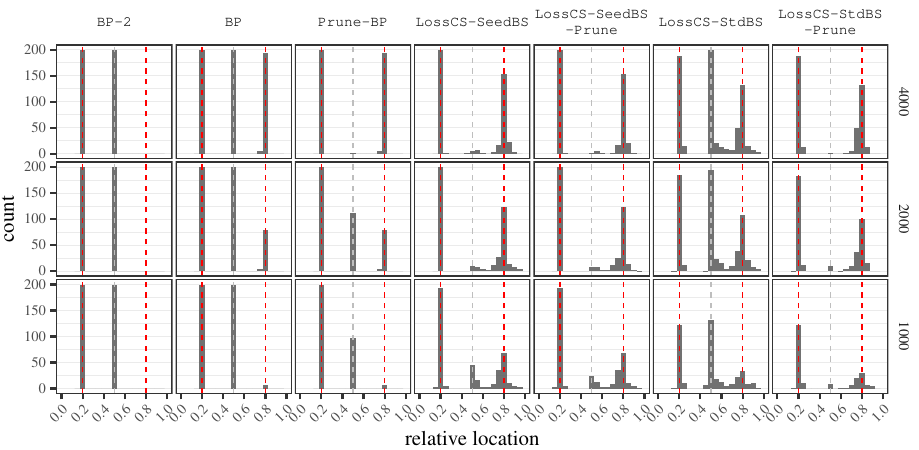}
\caption{Histogram of the estimated CCPs using different methods. The two red
dashed lines correspond to the relative location of the true CCPs. The grey
dashed line corresponds to the relative location of the true NCCP. When the
sample size is large, \texttt{Prune-BP} performs best in terms of both the
number of false positives and true positives while at a relatively small sample
size \texttt{LossCS-SeedBS-Prune} performs best.}
\label{fig:exp_set3}
\end{figure}

In summary, both families of \texttt{LossCS-*} and \texttt{Prune-*} can be
helpful when localizing CCPs. The \texttt{LossCS-*} methods can be used in
combination with many existing change point localization schemes. The
\texttt{Prune-*} methods come with the usual guarantees of a test, which may be
beneficial for small sample sizes (where it may be better to remain conservative
and not make any decision). If the set of candidates is incorrect in that it
does not contain all true CCPs, the output of these methods may be incorrect;
however, in some scenarios it may be possible to realize that, see Experiment~4
in Appendix~\ref{app:add_exp}.

\section{Discussion} \label{sec:discussion} 

We introduce a notion of causal changes which, under additional causal
assumptions, captures changes in the underlying causal mechanism, but which is
still meaningful without referencing an underlying causal model. We consider the
problems of detecting and localizing these changes in an offline sequential
setting. For detection, we propose a testing procedure that combines several
invariance tests. For localization, we propose two approaches based on pruning a
set of candidate CCPs and based on minimizing a loss function, respectively. The
first approach \texttt{Prune-*} is directly applicable if a superset of all CCPs
is known. If this is not the case, one can first estimate the RCPs using a
method that can localize both changes in the regression parameter and changes in
the residual distribution. If the candidates are imprecise, an NCCP may be
mistaken as a CCP in the final estimates. The second approach \texttt{LossCS-*}
avoids this problem by targeting the CCPs directly without estimating the RCPs
first.

We have considered a sequential setting in this work, and one future direction
is to take time dependencies into consideration. Moreover, we have focused on
linear regression models, but the idea can potentially be generalized to
semi-parametric or non-parametric regressions. In those settings, we may rely on
a notion of invariant functions rather than invariant sets. 

\section*{Acknowledgement}

SH and NP are supported by a research grant (0069071) from Novo Nordisk Fonden.
We thank Solt Kovács for the helpful discussions and Rikke Søndergaard for her
contribution in the initial stage of this work. Part of the work was done while
JP was at the University of Copenhagen.

%% file: tikz_split_new.tex
\resizebox*{0.6\textwidth}{!}{
\begin{tikzpicture}
\draw[thick] (0, 0) -- (10,0); 
\foreach \x in {0, 3, 6, 10} 
\draw[ultra thick] (\x cm, 5pt) -- (\x cm, -5pt); 
\draw (0, 0) node[below=5pt] {$\min(I)$}; 
\draw (3, 0) node[below=5pt] (k1) {$\min(I) + s$}; 
\draw (6, 0) node[below=5pt] (k2) {$\min(I) + 2s$}; 
\draw (10, 0) node[below=5pt] {$\max(I)$}; 
\draw (1.5, 0) node[below=5pt] (t1) {}; 
\node[above = 0.3cm of t1]
{$\overbrace{\phantom{\rule{0.18\textwidth}{0.5pt}}}$};
\node[above = 0.7cm of t1]{$P_1(I)$};
\draw (4.5, 0) node[below=5pt] (t2) {}; 
\node[above = 0.3cm of t2]
{$\overbrace{\phantom{\rule{0.18\textwidth}{0.5pt}}}$};
\node[above = 0.7cm of t2]{$P_2(I)$};
\draw (8, 0) node[below=5pt] (t3) {}; 
\node[above = 0.3cm of t3]
{$\overbrace{\phantom{\rule{0.24\textwidth}{0.5pt}}}$};
\node[above = 0.7cm of t3]{$P_3(I)$};
\end{tikzpicture}}

%% file: appendix.tex
\appendix
\renewcommand{\thesection}{\Alph{section}}
\counterwithin{figure}{section}
\renewcommand\thefigure{\thesection\arabic{figure}}
\counterwithin{table}{section}
\renewcommand\thetable{\thesection\arabic{table}}

\begin{itemize}
\item Section~\ref{app:ex_details}: Additional examples and details on examples.
\item Section~\ref{app:alg}: Algorithms
\item Section~\ref{app:add_exp_and_det}: Additional numerical experiments and
experiment details
\item Section~\ref{app:proofs}: Proofs
\item Section~\ref{app:aux_res}: Auxiliary results
\end{itemize}

\section{Additional examples and details on examples} \label{app:ex_details}

\subsection{Details of Example~\ref{ex:lscm}}

For all $j\in\{X^1,X^2,X^3,Y\}$, denote the variance of $\epsilon_i^j$ as
$(\sigma_i^j)^2$. For all $i \in\{1,\ldots,n\}$, $\alpha = (0,0,1)^\top$; for
all $i\in I_1 \cup I_2$, $A_i = \mathbf{0}_{3\times 4}$ and $\beta_i =
(1,1,0,0)^\top$, and for all $i \in I_3$, $A_i = \begin{pmatrix} 0 & 0 & 0 & 0
\\
0 & 0 & 0 & 0\\
1 & 0 & 0 & 0\end{pmatrix}$ and $\beta_i = (0,1,0,0)^\top$. For all $i\in I_1$
and for all $j\in\{X^1,X^2,X^3,Y\}$, $\sigma_i^j = 1$, and for all $i \in
I_2\cup I_3$, $\sigma_i^{X^1} = \sigma_i^{Y} = 1$ and $\sigma_i^{X^2} =
\sigma_i^{X^3} = 2$. This means that the population OLS coefficient for all
$i\in I_1$ is $\beta^{\text{OLS}} = (0.5, 0.5, 0.5, 0)^\top$, for all $i\in I_2$
is $\beta^{\text{OLS}} = (0.8, 0.8, 0.2, 0)^\top$, and for all $i\in I_3$ is
$\beta^{\text{OLS}} = (-0.2, 0.8, 0.2, 0)^\top$.

\subsection{Example where the reverse implication of
Proposition~\ref{prop:ccp_under_seq_lscm_with_hidden} does not hold}
\label{app:ex_rev_fail}

For all $i\in\{1,\ldots,n\}$ consider an SCM over $(X_i, Y_i)$ given by $X_i^{3}
\coloneqq 1$ as intercept and 
\begin{align*}
  X_i^1 &\coloneqq \epsilon_i^{X^1}, \\
  Y_i &\coloneqq \beta_i X_i^1 + \epsilon_i^Y, \\
  X_i^2 &\coloneqq \alpha_i Y_i + \epsilon_i^{X^2},
\end{align*}
where $\epsilon_i^j \sim\cN\left(0,(\sigma_i^j)^2\right)$ for all $j\in\{X^1,
X^2, Y\}$. Let $k\in\{1,\ldots,n\}$ with $1 < k < n$. For all
$i\in\{1,\ldots,k-1\}$, $\beta_i = 2$, $\alpha_i = 3$, and for all $j\in\{X^1,
X^2, Y\}$, $\sigma_i^j = 1$. For all $i\in\{k,\ldots,n\}$, $\beta_i = 1$,
$\alpha_i = 8/3$, $\sigma_i^{X^1} = 1$, $\sigma_i^{X^2} = \sqrt{2}$, and
$\sigma_i^{Y} = 3\sqrt{2}/4$ . Then, for all $i\in\{1,\ldots,n\}$,
$\hat\beta_i^{\text{OLS}} = (0.2, 0.3,0)^\top$. Thus, $k$ is not a CCP even
though $\beta_k\neq\beta_{k-1}$ and $\epsilon_k^Y\dneq\epsilon_{k-1}^Y$.

\section{Algorithms} \label{app:alg}

The procedure of testing candidates described in Section~\ref{sec:ccpl_candi} is
given in Algorithm~\ref{alg:prune_candi}. Other algorithms mentioned in
Section~\ref{sec:ccpl}: the standard binary segmentation
\citep{vostrikova1981detecting} is given in Algorithm~\ref{alg:bin_seg}, the
construction of seeded intervals is given in
Algorithm~\ref{alg:seeded_intervals}, and the seeded binary segmentation is
given in Algorithm~\ref{alg:seededBS}.

\begin{algorithm}[!hbt]
\caption{Pruning: CCP localization given candidates} \label{alg:prune_candi}
\DontPrintSemicolon
\SetKwInput{KwIn}{input}
\SetKwInput{KwOut}{output}
\KwIn{data $(\bX, \bY)$, CCP candidates $\{k_1,\ldots,k_L\}$ with $k_i < k_j$ 
if $i < j$, and tests $(\varphi_I)_{I\in\cI}$ for $(\cH_0^I)_{I\in\cI}$}
\BlankLine
$k_0 \gets 1$; $k_{L+1} \gets n+1$; $\widehat\cT\gets \emptyset$\;
\For{$\ell \in \{1,\ldots,L\}$}{
  $I \gets \{k_{\ell-1},\ldots, k_{\ell+1}-1\}$\;
  \If{$\varphi_I(\bX_I,\bY_I) = 1$}{
    $\widehat\cT\gets \widehat\cT \cup \{k_\ell\}$\;
  }
}
\KwOut{$\widehat\cT$}
\end{algorithm}

\begin{algorithm}
\caption{Binary segmentation}
\label{alg:bin_seg}
\DontPrintSemicolon
\SetKwInput{KwTune}{\normalfont \textit{tuning}}
\SetKwFunction{FMain}{BinSeg}
\SetKwProg{Pn}{function}{:}{ }
\SetKwInput{KwIn}{input}
\SetKwInput{KwOut}{output}
\KwIn{data $(\bX, \bY)$, tests $(\varphi_I)_{I\in\cI}$ for
$(\cH_0^I)_{I\in\cI}$}
\KwTune{a minimal segmentation length $s$}
\Pn{\FMain{$\bX, \bY, b, e, s$}}{
  Let $I \coloneqq \{b,\ldots, e\}$ \;
  \eIf{$|I| > h$ and $\varphi_I = 1$}{
    $k \coloneqq \displaystyle\argmin_{i\in I} \widehat\cL_{I,s}(i)$
    \;
    $G \coloneqq$ \FMain{$\bX, \bY, b, k-1, s$} \; 
    $D \coloneqq$ \FMain{$\bX, \bY, k, e, s$} \; 
    \KwRet{$G\cup\{d\}\cup D$} \;
  }{
    \KwRet{$\emptyset$}\;
  }
}
Let $\widehat{\cD}\coloneqq$ \FMain{$\bX,\bY, 1, n, s$}. \;
\KwOut{$\widehat\cD$}
\end{algorithm}

\begin{algorithm}
\caption{Seeded intervals} \label{alg:seeded_intervals}
\DontPrintSemicolon
\SetKwInput{KwTune}{\normalfont \textit{tuning}}
\SetKwInput{KwIn}{input}
\SetKwInput{KwOut}{output}
\KwIn{Total number of time points $n$}
\KwTune{a decay number $a \in [1/2, 1)$, and a minimal segmentation length $s$}
Let $\mathfrak{L}_1 \coloneqq \{\{1,\ldots,n\}\}$ and $L \coloneqq \lfloor 1 +
\log_{\frac{1}{a}}\frac{n}{s}\rfloor$ \tcp*{$L$ is the number of levels}
\For{$\ell\in\{2,\ldots,L\}$}{
  Let $h_\ell \coloneqq na^{\ell-1}$, 
  $s_\ell \coloneqq \frac{n-h_\ell}{q_\ell-1}$, and $q_\ell \coloneqq 2\lceil\frac{1}{a^{\ell-1}}\rceil$ \;
  \For{$j \in \{1,\ldots, q_\ell\}$}{
    $I_j \coloneqq \big\{\lfloor(j-1)s_\ell\rfloor+1, \ldots,\lceil(j-1)s_\ell+h_\ell\rceil\big\}$
  }
  $\fL_\ell \coloneqq 
\displaystyle\big\{I_1,\ldots,I_{q_\ell}\big\}$ \tcp*{$\fL_\ell$ denotes level
$\ell$, which is a set of intervals}
}
\KwOut{$L$ levels of intervals $\fL_1,\cdots,\fL_L$}
\end{algorithm}

\begin{algorithm}
\caption{Seeded binary segmentation (with narrowest-over-threshold)}
\label{alg:seededBS}
\DontPrintSemicolon
\SetKwInput{KwTune}{\normalfont \textit{tuning}}
\SetKwInput{KwIn}{input}
\SetKwInput{KwOut}{output}
\KwIn{data $(\bX, \bY)$, tests $(\varphi_I)_{I\in\cI}$ for
$(\cH_0^I)_{I\in\cI}$, and levels of intervals $\fL_1,\ldots,\fL_L$ obtained from
Algorithm~\ref{alg:seeded_intervals}}
Let $\widehat\cT \coloneqq \emptyset$ and $i \coloneqq 0$\;
\While{$i < L$} {
  \eIf{there exists $I\in\fL_{L-i}$ such that $\varphi_I = 1$}{
    Let $I$ be the interval in $\fL_{L-i}$ that has the smallest p-value \;
    Let $k \coloneqq \argmin_{i\in I} \widehat\cL_{I,s}(i)$ and $\widehat\cT \coloneqq
\widehat\cT\cup\{k\}$ \;
    Update $\fL_1,\ldots,\fL_L$ by removing all intervals in all $\fL_1,\ldots,\fL_L$ that contain $k$
  }{
    Remove $\fL_i$
  }
  $i \coloneqq i + 1$
}
\KwOut{$\widehat\cT$}
\end{algorithm}

\section{Additional numerical experiments and experiment details}
\label{app:add_exp_and_det}

\subsection{Additional numerical experiments} \label{app:add_exp}

Further details on the following experiments can be found in
Appendix~\ref{app:exp_det}.

\paragraph{Experiment 4: Pruning when the candidate set is a subset of the true
CCPs} For a total number of time points $n$, two CCPs are fixed at $\lceil0.25 n
+ 1\rceil$ and $\lceil0.75 n +1\rceil$. We illustrate the performance of pruning
described in Section~\ref{sec:ccpl_candi} when the candidate set contains only
one of the two CCPs. Figure~\ref{fig:exp_set4} shows the power of detecting all
CCPs contained in the candidate set when the candidate set contains only the
first CCP ($k_1$), only the second CCP ($k_2$), or contains both CCPs ($k_1,
k_2$). In this setting, we observe that when the candidate set contains only one
of the two true CCPs, the detection of this CCP is not affected for a large
enough sample size.  

\begin{figure}[ht]
\centering
\includegraphics[width=0.6\textwidth]{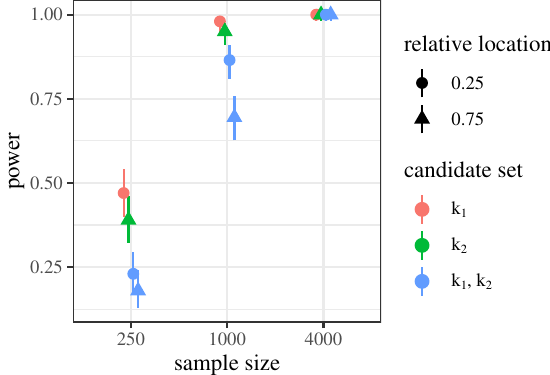}
\caption{Power of detecting all CCPs contained in the candidate set when the
candidate set contains only one true CCP or both of the true CCPs.}
\label{fig:exp_set4}
\end{figure}

\paragraph{Experiment 5: Pruning when the candidate deviates from the true CCP}
For $n = 4000$, one true CCP is fixed at $0.5 n$. We let the candidate equal
$(0.5 + \delta) n$ for different $\delta$. Figure~\ref{fig:exp_set5} shows (with
the marker `$\times$') the percentage of repetitions for which the method
(incorrectly for $\delta \neq 0$) outputs the candidate as a CCP, i.e., the test
that the candidate is not a CCP is rejected (with $\alpha = 0.05$). If the
method does output the candidate as a CCP, we test whether there exists an
invariant set over the sub-intervals to the left and right of the candidate. The
percentages of repetitions where the test is rejected in at least one of the
sub-intervals versus not rejected in both sub-intervals given different
$\delta$'s are also shown with bar charts in Figure~\ref{fig:exp_set5}. Even
though an inaccurate candidate can be output as an estimated CCP, it is possible
to invalidate the output by testing for the existence of invariant sets over the
sub-intervals divided by the candidate.

\begin{figure}[ht]
\centering
\includegraphics[width=0.8\textwidth]{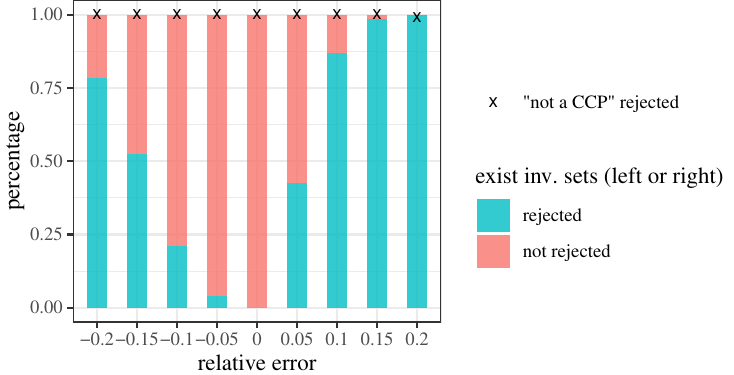}
\caption{When the candidate deviates from the true CCP by a relative error
$\delta \in \{-0.2,-0.15,\ldots,0.2\}$, the percentages of cases where `the
candidate is not a CCP' is rejected, are almost all at $1$ (when $\delta = 0$
the rejection is correct while when $\delta\neq 0$ the rejection is incorrect).
However, we can test whether there is an invariant set in the two sub-intervals
to the left or right of the candidate that is output as a CCP. If we reject that
there is an invariant set in one of the two sub-intervals to the left or right
of this candidate, it indicates that classifying the candidate as a CCP may be
incorrect.}
\label{fig:exp_set5}
\end{figure}

\subsection{Experiment details}  \label{app:exp_det}

For each of the experiments, we first choose the parameters of the SCMs
\eqref{eq:five_node_mb_graph} 
for all $i\in\{1,\ldots,n\}$, and then generate one dataset accordingly for each
of the $200$ repetitions. As before, the parameters are chosen such that changes
in the causal parameters are CCPs and changes in the non-causal parameters are
NCCPs (we have verified this using straight-forward computations). The specific
values of these parameters for each experiment in Section~\ref{sec:exp_ccpl} are
given below.

\begin{itemize}
\item \emph{CCP detection in Section~\ref{sec:exp_ccpd}}: Table~\ref{tab:param_exp_ccpd}.
\item \emph{CCP localization Experiment 1 in Section~\ref{sec:exp_ccpl}}:
Table~\ref{tab:param_exp_ccpl1}.
\item \emph{CCP localization Experiment 2 in Section~\ref{sec:exp_ccpl}}: Table~\ref{tab:param_exp_ccpl2}.
\item \emph{CCP localization Experiment 3 in Section~\ref{sec:exp_ccpl}}: Table~\ref{tab:param_exp_ccpl3}.
\item \emph{CCP localization Experiment 4 in Appendix~\ref{app:add_exp}}: Table~\ref{tab:param_exp_ccpl4}.
\item \emph{CCP localization Experiment 5 in Appendix~\ref{app:add_exp}}:
Table~\ref{tab:param_exp_ccpl5}.
\end{itemize}

\begin{table}[!htb]
\centering
\resizebox{\textwidth}{!}{
\begin{tabular}{rrrrrrrrrrrrrrrr}
\toprule
$i$ & $\mu_i^1$ & $\mu_i^2$ & $\mu_i^3$ & $\mu_i^4$ & $\mu_i^Y$ & $\sigma_i^1$ &
$\sigma_i^2$ & $\sigma_i^3$ & $\sigma_i^4$ & $\sigma_i^Y$ & $a_i^{12}$ &
$a_i^{53}$ & $a_i^{43}$ & $b_i^{15}$ & $b_i^{25}$ \\ 
\midrule
$1,\ldots, n\nu$ & 1.00 & 1.00 & 1.00 & 1.00 & 1.00 & 1.00 & 1.00 & 1.00 &
1.00 & 1.00 & 1.00 & 1.00 & 1.00 & 1.00 & 1.00 \\ 
$n\nu+1,\ldots, n$ & 1.00 & 1.00 & 1.00 & 1.00 & 1.00 & 1.00 & 1.00 & 1.00 &
1.00 & 1.00 & 1.00 & 1.00 & 1.00 & 2.00 & 2.00 \\ 
\bottomrule
\end{tabular}} \caption{Parameter values of CCP detection experiment in
Section~\ref{sec:exp_ccpd} where $\nu = \displaystyle\frac{\tau - 1}{n}$ is the
relative location of the single CCP, $\nu\in\{0.1,\ldots,0.9\}$.}
\label{tab:param_exp_ccpd}
\end{table}

\begin{table}[!htb]
\centering
\resizebox{\textwidth}{!}{
\begin{tabular}{rrrrrrrrrrrrrrrr}
\toprule
$i$ & $\mu_i^1$ & $\mu_i^2$ & $\mu_i^3$ & $\mu_i^4$ & $\mu_i^Y$ & $\sigma_i^1$ &
$\sigma_i^2$ & $\sigma_i^3$ & $\sigma_i^4$ & $\sigma_i^Y$ & $a_i^{12}$ &
$a_i^{53}$ & $a_i^{43}$ & $b_i^{15}$ & $b_i^{25}$ \\ 
\midrule
$1,\ldots,0.25 n$ & 1.00 & 1.00 & 1.00 & 1.00 & 1.00 & 1.00 & 1.00 & 1.00 & 1.00 & 1.00 & 1.00 & 1.00 & 1.00 & 1.00 & 1.00 \\ 
$0.25 n+1, \ldots, 0.5 n$ & 1.50 & 0.50 & 0.50 & 1.50 & 1.00 & 0.71 & 0.71 & 1.22 & 1.22 & 1.00 & 1.50 & 1.50 & 0.50 & 1.00 & 1.00 \\ 
$0.5 n + 1, \ldots, 0.75 n$ & 1.50 & 0.50 & 0.50 & 1.50 & 0.50 & 0.71 & 0.71 &
1.22 & 1.22 & 1.22 & 1.50 & 1.50 & 0.50 & 1.50 & 0.50 \\ 
$0.75 n + 1,\ldots, n$ & 0.75 & 0.75 & 0.25 & 0.75 & 0.50 & 0.50 & 0.50 & 1.50 & 0.87 & 1.22 & 2.25 & 0.75 & 0.25 & 1.50 & 0.50 \\ 
\bottomrule
\end{tabular}} \caption{Parameter values of CCP localization Experiment 1 in
Section~\ref{sec:exp_ccpl} where $0.5 n$ is the location of the single CCP, and
$0.25 n$ and $0.75 n$ are the locations of two NCCPs. The changes at the CCP and
the NCCPs are such that each causal (respectively, non-causal) parameter either
increase or decrease (with probability $0.5$) by $50\%$ of its value at the
previous time point (for $\sigma_i^j$ where $j\in\{1,2,3,Y\}$, the changes are
$50\%$ in their squared values).}
\label{tab:param_exp_ccpl1}
\end{table}

\begin{table}[!htb]
\centering
\resizebox{\textwidth}{!}{
\begin{tabular}{rrrrrrrrrrrrrrrr}
\toprule
$i$ & $\mu_i^1$ & $\mu_i^2$ & $\mu_i^3$ & $\mu_i^4$ & $\mu_i^Y$ & $\sigma_i^1$ &
$\sigma_i^2$ & $\sigma_i^3$ & $\sigma_i^4$ & $\sigma_i^Y$ & $a_i^{12}$ &
$a_i^{53}$ & $a_i^{43}$ & $b_i^{15}$ & $b_i^{25}$ \\ 
\midrule
$1,\ldots,200$ & 1.00 & 1.00 & 1.00 & 1.00 & 1.00 & 1.00 & 1.00 & 1.00 & 1.00 & 1.00 & 1.00 & 1.00 & 1.00 & 1.00 & 1.00 \\ 
$201,\ldots,500$ & 1.00 & 1.00 & 1.00 & 1.00 & 1.50 & 1.00 & 1.00 & 1.00 & 1.00 & 1.22 & 1.00 & 1.00 & 1.00 & 1.50 & 1.50 \\ 
$501,\ldots,1500$ & 1.50 & 1.50 & 1.50 & 1.50 & 1.50 & 1.22 & 1.22 & 1.22 & 1.22 & 1.22 & 1.50 & 1.50 & 1.50 & 1.50 & 1.50 \\ 
$1501,\ldots,2000$ & 2.25 & 2.25 & 2.25 & 2.25 & 1.50 & 1.50 & 1.50 & 1.50 &
1.50 & 1.22 & 2.25 & 2.25 & 2.25 & 1.50 & 1.50 \\
\bottomrule
\end{tabular}} \caption{An example of the parameter values of CCP localization
experiment in Section~\ref{sec:exp_ccpl} with $\nu = 0.1$. In this experiment,
$n=2000$, two fixed NCCPs are located at $501$ and $1501$, and one CCP is
located at $\tau = n\nu + 1$. The changes at the CCP and the NCCPs are such that
each causal parameter (respectively, non-causal) increases by $50\%$ of its
value at the previous time point (for $\sigma_i^j$ where $j\in\{1,2,3,Y\}$, the
changes are $50\%$ in their squared values). The parameter values for $\nu
\in\{0.2, \ldots, 0.9\}$ are constructed in the same way.}
\label{tab:param_exp_ccpl2}
\end{table}

\begin{table}[!htb]
\centering
\resizebox{\textwidth}{!}{
\begin{tabular}{rrrrrrrrrrrrrrrr}
\toprule
$i$ & $\mu_i^1$ & $\mu_i^2$ & $\mu_i^3$ & $\mu_i^4$ & $\mu_i^Y$ & $\sigma_i^1$ &
$\sigma_i^2$ & $\sigma_i^3$ & $\sigma_i^4$ & $\sigma_i^Y$ & $a_i^{12}$ &
$a_i^{53}$ & $a_i^{43}$ & $b_i^{15}$ & $b_i^{25}$ \\ 
\midrule
$1,\ldots,0.25 n$ & 1.00 & 1.00 & 1.00 & 1.00 & 1.00 & 1.00 & 1.00 & 1.00 & 1.00
& 1.00 & 1.00 & 1.00 & 1.00 & 1.00 & 1.00 \\ 
$0.25 n+1,\ldots,0.5 n$ & 1.00 & 1.00 & 1.00 & 1.00 & 1.00 & 1.00 & 1.00 & 1.00
& 1.00 & 1.00 & 1.00 & 1.00 & 1.00 & 1.00 & 0.00 \\ 
$0.5 n+1,\ldots, 0.75 n$ & 1.00 & 1.00 & 2.00 & 2.00 & 1.00 & 1.00 & 1.00 & 1.00 & 1.00 & 1.00 & 1.00 & 2.00 & 1.00 & 1.00 & 0.00 \\ 
$0.75 n+1,\ldots, n$ & 1.00 & 1.00 & 2.00 & 2.00 & 1.00 & 1.00 & 1.00 & 1.00 &
1.00 & 2.00 & 1.00 & 2.00 & 1.00 & 0.00 & 1.00 \\ 
\bottomrule
\end{tabular}} \caption{Parameter values of CCP localization experiment
Experiment 1 in Section~\ref{sec:exp_ccpl} where $0.5 n$ is the location of the
single CCP, and $0.25 n$ and $0.75 n$ are the locations of two NCCPs.}
\label{tab:param_exp_ccpl3}
\end{table}

\begin{table}[!htb]
\centering
\resizebox{\textwidth}{!}{
\begin{tabular}{rrrrrrrrrrrrrrrr}
\toprule
$i$ & $\mu_i^1$ & $\mu_i^2$ & $\mu_i^3$ & $\mu_i^4$ & $\mu_i^Y$ & $\sigma_i^1$ &
$\sigma_i^2$ & $\sigma_i^3$ & $\sigma_i^4$ & $\sigma_i^Y$ & $a_i^{12}$ &
$a_i^{53}$ & $a_i^{43}$ & $b_i^{15}$ & $b_i^{25}$ \\ 
\midrule
$1,\ldots,0.25 n$ & 1.00 & 1.00 & 1.00 & 1.00 & 1.00 & 1.00 & 1.00 & 1.00 & 1.00
& 1.00 & 1.00 & 1.00 & 1.00 & 1.00 & 1.00 \\ 
$0.25 n+1,\ldots,0.75 n$ & 1.00 & 1.00 & 1.00 & 1.00 & 1.50 & 1.00 & 1.00 & 1.00
& 1.00 & 1.22 & 1.00 & 1.00 & 1.00 & 1.50 & 1.50 \\ 
$0.75 n+1,\ldots, n$ & 1.00 & 1.00 & 1.00 & 1.00 & 2.25 & 1.00 & 1.00 & 1.00 &
1.00 & 1.50 & 1.00 & 1.00 & 1.00 & 2.25 & 2.25 \\ 
\bottomrule
\end{tabular}} \caption{Parameter values of CCP localization experiment Experiment 4 in
Appendix~\ref{app:add_exp} where $0.25 n$ and $0.75 n$ are the locations of two
CCPs. The changes at the CCP are such that each causal parameter
increases by $50\%$ of its value at the previous time point (for $\sigma_i^Y$
the changes are $50\%$ in its squared value).}
\label{tab:param_exp_ccpl4}
\end{table}

\begin{table}[!htb]
\centering
\resizebox{\textwidth}{!}{
\begin{tabular}{rrrrrrrrrrrrrrrr}
\toprule
$i$ & $\mu_i^1$ & $\mu_i^2$ & $\mu_i^3$ & $\mu_i^4$ & $\mu_i^Y$ & $\sigma_i^1$ &
$\sigma_i^2$ & $\sigma_i^3$ & $\sigma_i^4$ & $\sigma_i^Y$ & $a_i^{12}$ &
$a_i^{53}$ & $a_i^{43}$ & $b_i^{15}$ & $b_i^{25}$ \\ 
\midrule
$1,\ldots, 2000$ & 1.00 & 1.00 & 1.00 & 1.00 & 1.00 & 1.00 & 1.00 & 1.00 & 1.00 & 1.00 & 1.00 & 1.00 & 1.00 & 1.00 & 1.00 \\ 
$2000,\ldots,4000$ & 1.00 & 1.00 & 1.00 & 1.00 & 1.50 & 1.00 & 1.00 & 1.00 &
1.00 & 1.22 & 1.00 & 1.00 & 1.00 & 1.50 & 1.50 \\   
\bottomrule
\end{tabular}} \caption{Parameter values of CCP localization experiment
Experiment 5 in Appendix~\ref{app:add_exp} where $2000$ is the location of the
single CCP. The changes at the CCP are such that each causal parameter
increases by $50\%$ of its value at the previous time point (for $\sigma_i^Y$
the change is $50\%$ in its squared value).}
\label{tab:param_exp_ccpl5}
\end{table}

\FloatBarrier

\section{Proofs} \label{app:proofs}

\subsection{Proof of Proposition~\ref{prop:ccp_alt_char}}

\begin{proof}
We show both directions separately.
    
($\Rightarrow$) Assume $k$ is a CCP and fix an arbitrary $S \in \cS$. Then,
since $k$ is a CCP it holds that either
$\beta_k^{\operatorname{OLS}}(S)\neq\beta_{k-1}^{\operatorname{OLS}}(S)$ or
$\epsilon_{k}(S) \dneq \epsilon_{k-1}(S)$. This implies, that $S$ is not a
$\{k-1,k\}$-invariant set. Since $S$ was arbitrary this implies the result.

($\Leftarrow$) Assume there does not exist a $\{k-1,k\}$-invariant set $S
\in\cS$. Then, it holds for all $S\in\cS$ that either
$\beta_k^{\operatorname{OLS}}(S)\neq\beta_{k-1}^{\operatorname{OLS}}(S)$ or
$\epsilon_{k}(S) \dneq \epsilon_{k-1}(S)$, which implies that $k$ is a CCP.

This completes the proof of Proposition~\ref{prop:ccp_alt_char}.
\end{proof}

\subsection{Proof of Proposition~\ref{prop:ccp_under_seq_lscm_with_hidden}}

\begin{proof}
We begin by connecting the causal coefficient and residual with the
corresponding population OLS coefficient and residual given the parents of $Y$.
By the definitions of the causal and population OLS coefficients it holds for
all $i\in\{k-1,k\}$ and all
$j\in\{1,\ldots,d+1\}\setminus\operatorname{PA}(Y_i)$, that
$(\beta_i^{\text{OLS}}(\operatorname{PA}(Y_{i})))^j = \beta_i^j = 0$ and
\begin{align}
\Big(\beta_i^{\text{OLS}}(\operatorname{PA}(Y_{i}))\Big)^{\operatorname{PA}(Y_i)}
&= \bbE\Big[X_i^{\operatorname{PA}(Y_{i})}(X_i^{\operatorname{PA}(Y_{i})})^\top\Big]^{-1}\bbE\Big[X_i^{\operatorname{PA}(Y_{i})}Y_i\Big]\nonumber\\
&= \bbE\Big[X_i^{\operatorname{PA}(Y_{i})}(X_i^{\operatorname{PA}(Y_{i})})^\top\Big]^{-1}\bbE\Big[X_i^{\operatorname{PA}(Y_{i})}(X_i^{\operatorname{PA}(Y_{i})})^\top\beta_i^{\operatorname{PA}(Y_{i})} \nonumber \\ 
&\qquad\qquad\qquad\qquad\qquad\qquad + X_i^{\operatorname{PA}(Y_{i})}g_i(H_i, \epsilon_i^Y)\Big] \nonumber\\
&= \beta_i^{\operatorname{PA}(Y_{i})} + \bbE\Big[X_i^{\operatorname{PA}(Y_{i})}(X_i^{\operatorname{PA}(Y_{i})})^\top\Big]^{-1}\bbE\Big[X_i^{\operatorname{PA}(Y_{i})}g_i(H_i, \epsilon_i^Y)\Big] \nonumber\\
&= \beta_i^{\operatorname{PA}(Y_{i})},
\end{align}
where the last equality follows from the assumption that for all
$i\in\{k-1,k\}$, $\bbE[X_i^{\operatorname{PA}(Y_{i})}g_i(H_i, \epsilon_i^Y)] =
0$. Thus, we get for all $i\in\{k-1,k\}$ that
\begin{equation}
\label{eq:ols_causal_coefficients}
\beta_i^{\text{OLS}}(\operatorname{PA}(Y_i)) = \beta_i.
\end{equation}
Moreover, using this result and the definition of the population OLS residual
given $\operatorname{PA}(Y_i)$ we also obtain that
\begin{align}
\epsilon_i(\operatorname{PA}(Y_{i}))
&= Y_i-X_i^{\top}\beta_i^{\text{OLS}}(\operatorname{PA}(Y_{i}))\nonumber\\
&= Y_i-X_i^{\top}\beta_i\nonumber\\
&= g_i(H_i,\epsilon_i^Y).\label{eq:ols_causal_residuals}
\end{align}

Now, we can prove the result. To this end, assume $k\in\{2,\ldots,n\}$ satisfies
that $\beta_k=\beta_{k-1}$ and $g_k(H_k,\epsilon_k^Y)\deq
g_{k-1}(H_{k-1},\epsilon_{k-1}^Y)$. This implies that
$\operatorname{PA}(Y_{k-1})=\operatorname{PA}(Y_{k})$ and hence, using
\eqref{eq:ols_causal_coefficients} and \eqref{eq:ols_causal_residuals}, the set
$\operatorname{PA}(Y_{k})$ is $\{k-1, k\}$-invariant. Therefore, by
Proposition~\ref{prop:ccp_alt_char}, $k$ is not a CCP. This completes the proof
of Proposition~\ref{prop:ccp_under_seq_lscm_with_hidden}.
\end{proof}

\subsection{Proof of Proposition~\ref{prop:correct_level}}

\begin{proof}
Suppose $(\bX_I, \bY_I) \sim P \in \cH_0^I$ and let $S\in\cS$ be s.t.\ $S$ is
$I$-invariant. We then have
\begin{align*}
\bbP_P(\varphi_I = 1) \leq \bbP_P(\varphi_I^{S} = 1) 
\leq \alpha.
\end{align*}
\end{proof}

\subsection{Proof of Proposition~\ref{prop:candi_prop}}

\begin{proof}
For (i), assume that for all $k\in\cK$ the test $\varphi_{I_k}$ is level
$\alpha$. Then, using a union bound we get
\begin{align*}
\mathbb{P}(\widehat\cT\subseteq\cT)
&=1-\mathbb{P}(\exists k\in\widehat\cT:\, k\not\in\cT) \\
&\geq 1-\sum_{k\in\cK\setminus\cT}\bbP(k\in\widehat\cT) \\
&\geq 1-\sum_{k\in\cK\setminus\cT}\bbP(\varphi_{I_k}=1) \\
&\geq 1-(|\cK| - |\cT|)\cdot\alpha,
\end{align*}
where in the last step we used that $\cH^{I_k}_0$ is true for all
$k\in\cK\setminus\cT$. 

Similarly, for (ii), assume that for all $\ell\in\{1,\dots,L\}$ with
$k_{\ell}\in\cT$ it holds that $\mathbb{P}(\varphi_{I_{k_{\ell}}}=1)\leq\beta$,
then
\begin{align*}
  \mathbb{P}(\widehat\cT\supseteq\cT)
  &=1-\mathbb{P}(\exists k\in\cT:\, k\not\in\widehat\cT) \\
  &\geq 1-\sum_{k\in\cT}\mathbb{P}(k\not\in\widehat\cT) \\
  &= 1-\sum_{k\in\cT}\mathbb{P}(\varphi_{I_{k}}=0) \\    
  &\geq 1-|\cT|\cdot(1-\beta).
\end{align*}

This concludes the proof. 
\end{proof}

\subsection{Proof of Proposition~\ref{prop:min_ols_ins_wo_ccp}} 

\begin{proof}
Suppose there is no CCP in $I$. Then, there exists a set $S\in\cS$, a parameter
$\beta\in\bbR^{d+1}$ and a distribution $F_\epsilon$ such that for all $i\in I$
it holds that $\beta_i^{\text{OLS}}(S) = \beta$ and $\epsilon_i(S) \iid
F_\epsilon$. Moreover, by Lemma~\ref{lem:ols_reduce} it holds for all
$J\subseteq I$ and all $i\in I$ that $\epsilon_i^J(S)\iid F_\epsilon$. Hence, we
get for all $J, J'\subseteq I$ that
\begin{align*}
  V_{J',J}(S)
  &= \frac{1}{|J'|}\sum_{\ell\in
    J'}\bbE\Big[\epsilon_\ell^J(S)^2\Big] \\
  &= \frac{1}{|J'|}\sum_{\ell\in J'}\bbE_{\nu\sim F_\epsilon}\Big[\nu^2\Big] \\
  &=\bbE_{\nu\sim F_\epsilon}\Big[\nu^2\Big].
\end{align*}
Therefore, for $s\in\bbN$, we get
\begin{align*}
  \cC_s(I)
  &= \min_{\tilde{S}\in\cS} \sum_{r=1}^{m_s(I)} \Big(V_{P_r^c(I), P_r(I)}(\tilde{S}) - V_{P_r(I), P_r(I)}(\tilde{S})\Big)^2 \\ 
  &\leq \sum_{r=1}^{m_s(I)} \Big(V_{P_r^c(I), P_r(I)}(S) - V_{P_r(I), P_r(I)}(S)\Big)^2 \\
  &= \sum_{r=1}^{m_s(I)} \left(\bbE_{\nu\sim F_\epsilon}\Big[\nu^2\Big]
    - \bbE_{\nu\sim F_\epsilon}\Big[\nu^2\Big]\right)  \\
  &= 0.
\end{align*}
Now, since $\cC_s(I) \geq 0$, we have that $\cC_s(I) = 0$. This completes the
proof of Proposition~\ref{prop:min_ols_ins_wo_ccp}. 
\end{proof}

\section{Auxiliary results} \label{app:aux_res}

\begin{lemma} \label{lem:ols_reduce} Let $I \in \cI$ and $S\in\cS$ is an
$I$-invariant set. Then, for all $J\subseteq I$ and all $i\in I$, it holds that
$\beta_J^{\text{OLS}}(S) = \beta_i^{\text{OLS}}(S)$, and that there exists a
distribution $F_\epsilon$ such that $\epsilon_i^J(S) = \epsilon_i(S) \iid
F_\epsilon$.
\end{lemma}

\begin{proof}
Since $S$ is $I$-invariant, there exists $\beta\in\bbR^{d+1}$ and distribution
$F_{\epsilon}$ on $\bbR$ such that for all $i\in I$ it holds that
$\beta_i^{\text{OLS}}(S) = \beta$ and $\epsilon_i(S)\iid F_\epsilon$. Moreover,
since the population OLS coefficient satisfies for all $i\in I$ that
$\bbE\left[X_i^SY_i\right]=\bbE\left[X_i^S(X_i^S)^{\top}\right]\beta_i^{\text{OLS}}(S)^S$,
it immediately follows for all $J\subseteq I$ that
\begin{align*}
  \Big(\beta_J^{\text{OLS}}(S)\Big)^S
  &= 
    \Big[\sum_{\ell\in
    J}\bbE[X_\ell^S(X_\ell^S)^\top]\Big]^{-1}\sum_{\ell\in
    J}\bbE[X_\ell^SY_\ell] \\
  &= 
    \Big[\sum_{\ell\in
    J}\bbE[X_\ell^S(X_\ell^S)^\top]\Big]^{-1}\sum_{\ell\in
    J}\bbE\left[X_\ell^S(X_\ell^S)^{\top}\right]\beta_\ell^{\text{OLS}}(S)^S
  \\
  &= 
    \Big[\sum_{\ell\in
    J}\bbE[X_\ell^S(X_\ell^S)^\top]\Big]^{-1}\sum_{\ell\in
    J}\bbE\left[X_\ell^S(X_\ell^S)^{\top}\right]\beta^S \\
  &=\beta^S.
\end{align*}
Since, for all $j\in\{1,\ldots,d+1\}\setminus S$ it also holds that
$\Big(\beta_J^{\text{OLS}}(S)\Big)^j=\beta^j=0$, we get that
$\beta_J^{\text{OLS}}(S)=\beta$. Moreover, this further implies for all $i\in I$
that
\begin{equation*}
  \epsilon_i^J(S) =
  Y_i-X_i^\top\beta_J^{\text{OLS}}(S) =
  Y_i-X_i^\top\beta =  Y_i-X_i^\top\beta_i^{\text{OLS}}(S) =
  \epsilon_i(S) \iid F_\epsilon,
\end{equation*}
this completes the proof of Lemma~\ref{lem:ols_reduce}.
\end{proof}

\section{Chow test} \label{app:chow_test}

Here we review the Chow test \citep{chow1960tests} where we adapted the setup
and notations to this paper.

Let $I = \{t_1,\ldots,t_l\} \in\mathcal{I}$ where $l > d$, and consider an
arbitrary $k \in I$, $I_k^1 \coloneqq \{t_1,\ldots,k\}$, and $I_k^2 \coloneqq
\{k+1,\ldots,t_l\}$ be the two non-overlapping subsets splitting $I$ at $k$.
Denote $l_1 \coloneqq |I_k^1|$ and $l_2 \coloneqq l - l_1$.

Assume that $l_1 > d$, and for all $m \in \{1,2\}$ and $i\in I_k^m$,
\begin{equation}
Y_i = X_i^S\beta_{m} + \epsilon_i \quad \text{and} \quad \bbE[\epsilon_i \mid
X_i^S] = 0
\end{equation}
with $\epsilon_i \iid \cN(0,\sigma^2)$. Then, the null hypothesis
\begin{equation}\label{hypo:chow}
\cH_0^S(I, k): \beta_1 = \beta_2
\end{equation} 
holds if \eqref{hypo:s_inv} holds. The Chow test \citep{chow1960tests} described
below can be used for testing~\eqref{hypo:chow}.

\begin{proposition}[Chow test] Let $\hat\beta_{I_k^1} \coloneqq
(\bX_{I_k^1}^{\top}\bX_{I_k^1})^{-1}\bX_{I_k^1}\bY_{I_k^1}$, $\hat\beta_{I_k^2}
\coloneqq (\bX_{I_k^2}^{\top}\bX_{I_k^2})^{-1}\bX_{I_k^2}\bY_{I_k^2}$, and
$\hat\beta_{I} \coloneqq (\bX_{I}^{\top}\bX_{I})^{-1}\bX_{I}\bY_{I}$. Denote the
residuals by $R_{I_k^1} \coloneqq \bY_{I_k^1} - \bX_{I_k^1}\hat\beta_{I_k^1}$
and $R_{I_k^2} \coloneqq \bY_{I_k^2} - \bX_{I_k^2}\hat\beta_{I_k^2}$. Then,
under the null hypothesis $\cH_0^S(I,k)$ the following two statements hold;
\begin{itemize}
\item if $l_2 > d$, 
\begin{equation} \label{eq:chow_eq29}
\frac{||\bX_{I_k^1}\hat\beta_{I_k^1} - \bX_{I_k^2}\hat\beta_{I}||^2 + ||\bX_{I_k^2}\hat\beta_{I_k^2} - \bX_{I_k^2}\hat\beta_{I}||^2}{||R_{I_k^1}||^2 + ||R_{I_k^2}||^2} \cdot \frac{l-2d}{d} \sim F(d, l-2d)
\end{equation}
\item if $l_2 \leq d$,
\begin{equation} \label{eq:chow_eq31}
\begin{split} 
&\frac{(\bY_{I_k^2} - \bX_{I_k^2}\hat\beta_{I_k^1})^\top[\bI_{l_2} + \bX_{I_k^2}(\bX_{I_k^1}^\top\bX_{I_k^1})^{-1}\bX_{I_k^2}^\top]^{-1}(\bY_{I_k^2} - \bX_{I_k^2}\hat\beta_{I_k^1})}{||R_{I_k^1}||^2 } \cdot \frac{l_1-d}{l_2} \\
=&\frac{||\bX_{I_k^1}\hat\beta_{I_k^1} - \bX_{I_k^2}\hat\beta_{I}||^2 + ||\bY_{I_k^2} - \bX_{I_k^2}\hat\beta_{I}||^2}{||R_{I_k^1}||^2 } \cdot \frac{l_1-d}{l_2} \\
\sim& F(l_2, l_1-d),
\end{split}
\end{equation}
where $\bI_l$ denotes the identity matrix of dimension $l$.

\end{itemize}
\end{proposition}

\begin{proof}

See \citet{chow1960tests}.

\end{proof}